\pdfoutput=1
\RequirePackage{ifpdf}
\ifpdf 
\documentclass[pdftex]{sigma}
\else
\documentclass{sigma}
\fi

\numberwithin{equation}{section}

\newtheorem{Theorem}{Theorem}[section]
\newtheorem{Corollary}[Theorem]{Corollary}
\newtheorem{Lemma}[Theorem]{Lemma}
\newtheorem{Proposition}[Theorem]{Proposition}
 { \theoremstyle{definition}
\newtheorem{Definition}[Theorem]{Definition}
\newtheorem{Example}[Theorem]{Example}
\newtheorem{Remark}[Theorem]{Remark} }

\def\nn{\nonumber}
\def\bs{\backslash}

\begin{document}
\allowdisplaybreaks

\newcommand{\arXivNumber}{2005.10288}

\renewcommand{\PaperNumber}{035}

\FirstPageHeading

\ShortArticleName{Functional Relations on Anisotropic Potts Models}

\ArticleName{Functional Relations on Anisotropic Potts Models:\\ from Biggs Formula to the Tetrahedron Equation}

\Author{Boris BYCHKOV~$^{\rm ab}$, Anton KAZAKOV~$^{\rm abc}$ and Dmitry TALALAEV~$^{\rm abc}$}

\AuthorNameForHeading{B.~Bychkov, A.~Kazakov and D.~Talalaev}

\Address{$^{\rm a)}$~Faculty of Mathematics, National Research University Higher School of Economics,\\
\hphantom{$^{\rm a)}$}~Usacheva 6, 119048, Moscow, Russia}
\EmailD{\href{mailto:bbychkov@hse.ru}{bbychkov@hse.ru}, \href{mailto:dtalalaev@yandex.ru}{dtalalaev@yandex.ru}, \href{mailto:anton.kazakov.4@mail.ru}{anton.kazakov.4@mail.ru}}

\Address{$^{\rm b)}$~Centre of Integrable Systems, P.G.~Demidov Yaroslavl State University,\\
\hphantom{$^{\rm b)}$}~Sovetskaya 14, 150003, Yaroslavl, Russia}

\Address{$^{\rm c)}$~Faculty of Mechanics and Mathematics, Moscow State University, 119991 Moscow, Russia}

\ArticleDates{Received July 06, 2020, in final form March 26, 2021; Published online April 07, 2021}

\Abstract{We explore several types of functional relations on the family of multivariate Tutte polynomials: the Biggs formula and the star-triangle ($Y-\Delta$) transformation at the critical point $n=2$. We~deduce the theorem of Matiyasevich and its inverse from the Biggs formula, and we apply this relation to construct the recursion on the parameter $n$. We~provide two different proofs of the Zamolodchikov tetrahedron equation satisfied by~the star-triangle transformation in the case of $n=2$ multivariate Tutte polynomial, we~extend the latter to the case of valency 2 points and show that the Biggs formula and the star-triangle transformation commute.}


\Classification{82B20; 16T25; 05C31}

\section{Introduction}
The theory of polynomial invariants of graphs in its current state uses many methods and tools of integrable statistical mechanics. This phenomenon demonstrates the inherent intrusion of mathematical physics methods into topology and combinatorics. In this paper, the main subject of research is functional relations in the family of polynomial invariants for framed graphs, in particular for multivariate Tutte polynomials~\cite{Sok}, their specializations for Potts models, multivariate chromatic and flow polynomials.

\looseness=1 The flow generating function is closely related to the problems of electrical networks on a~graph over a finite field. Each flow defines a discrete harmonic function, and non-zero flows can be interpreted as harmonic functions with a completely non-zero gradient. Specifically, we~discuss the full flow polynomial which is a linearization of the flow polynomial and, in particular, corresponds to the point of the compactification of the parameter space for the Biggs model.

One of the central tools of the paper is the Biggs formula (Lemma~\ref{L:Biggs}), which connects $n$-Potts models for different parameter values as a convolution with some weight over all edge subgraphs. In particular, we~offer a new proof of the theorem of Matiyasevich~\ref{T:Mat} about the connection of a flow and chromatic polynomial, as a special case of the Biggs formula. This interpretation allows us to construct an inverse statement of the theorem of Matiyasevich. Moreover, using the connection between the flow and the complete flow polynomial, we~obtain a shift of parameters in the Potts models (Theorem~\ref{T:sh}).

The fundamental type of correspondences on the space of the aforementioned invariants is the star-triangle type relations (also known as ``wye-delta'' relations) and the associated deletion-contraction relations. In a sense, the kinship between these relations is analogous to the role of~the tetrahedron equation in the local Yang--Baxter equation. Despite the fact that the invariance of the Ising model with respect to the star-triangle transformation is very well known~\cite{Bax}, we~have not found in the literature a full proof of the fact that the action of this transformation on the weights of an anisotropic system is a solution of the tetrahedron equation that corresponds to the orthogonal solution of the local Yang--Baxter equation: Theorem~\ref{maintheorem} (parts of this statement were mentioned in~\cite{Kash, Kor95, Serg2}). We~offer here two new proofs of this fact. We~find them instructive due to their anticipated relation to the theory of positive orthogonal grassmannians~\cite{HWX}.

The identification of the Potts model and the multivariate Tutte polynomial allows us to assert the existence of a critical point for the parameter $n$ in the family of Tutte polynomials. Namely, for $n=2$, this model has a groupoid symmetry generated by a family of transformations defined by the trigonometric solution of the Zamolodchikov tetrahedron equation. We~extend the star-triangle transformation for the graphs of lower valency in Section~\ref{sec:BST}. In this way, we~obtain a $14$-term correspondence. This extension commute with the Biggs formula. We~should mention the relation of this subject with the theory of cluster algebras. We~suppose that the multivariate Tutte polynomial on standard graphs at the critical point $n=2$ corresponds to the orthogonal version of the Lusztig variety~\cite{BFZ} in the case of the unipotent group and the electrical variety~\cite{GT} for the symplectic group.

\subsection{Organization of the paper}
In Section~\ref{sec:Biggs}, we~concentrate our attention on the Biggs formalism in the Ising and Potts type models. We~define the main recurrence relations and also identify the Tutte polynomial with the Potts model. Then, we~apply the Biggs formula to the proof of theorem of Matiyasevich and propose its inverse version. We~examine in details the recursion of the Potts model with respect to the parameter $n$.

In Section~\ref{sec:ST}, we~show that, if $n=2$, then the Potts model is invariant with respect to the star-triangle transformation given by the orthogonal solution for the local Yang--Baxter equation and the corresponding solution for the Zamolodchikov tetrahedron equation. In Section~\ref{sec:TE} we provide two different proofs for this fact. Both of them are interesting in the context of cluster variables on the space of Ising models. The first proof operates with the space of boundary measurement matrices and the second with the matrix of boundary partition function.

In Section~\ref{sec:BST}, we~show that the Biggs formula considered as a correspondence on the set of~mul\-tivariate Tutte polynomials commutes with the star-triangle transformation.

\section{Biggs interaction models}\label{sec:Biggs}
\subsection[n-Potts models and Tutte polynomial]{$\boldsymbol n$-Potts models and Tutte polynomial}

We define the anisotropic Biggs model (interaction model) on an undirected graph $G$ with the set of edges $E$ and the set of vertices $V$ (a graph can have multiple edges and loops) as follows:

\begin{itemize}\itemsep=0pt
\item a state $\sigma$ is a map $\sigma\colon V \rightarrow R$, where $R$ is a commutative ring with the unit,
\item the weight of the state $\sigma$ is defined by the formula
\begin{gather*}
 W_G(\sigma) := \prod\limits_{e\in E} i_e(\delta(e)),
\end{gather*}
 where $ \delta(e) = \sigma(v) - \sigma(w)$; the edge $e$ connects the vertices $v$ and $w$, the functions $i_e\colon R \rightarrow \mathbb {C} $ are even: $ \forall b \in R \colon i_e (b) = i_e (-b)$,
 \item the partition function $Z(G)$ of a model is the following sum
\begin{gather*}
Z(G)=\sum_{\sigma} W_G(\sigma),
\end{gather*}
where the summation is taken over all possible states $\sigma$.
 \end{itemize}

Let us consider the most simple Biggs interaction models:
 \begin{Definition} \label{def:nPotts}
 If $R \cong \mathbb{Z}_n $ and functions $i_e$ given as
 \begin{gather*}
 \begin{cases}
 i_e(0)=\alpha_e,\\
 i_e(a)=\beta_e, &\forall a \neq 0 \in R,
 \end{cases}
 \end{gather*}
 we call such model the anisotropic $n$-Potts model with the set of parameters $\alpha_e$ and $\beta_e$ and we denote it by $M(G;i_e)$.

 In addition, if the maps $i_e=i$ do not depend on edges, then we call such model the isotropic $n$-Potts model (or just $n$-Potts model) with parameters $\alpha$ and $\beta$. We~denote it by $M(G; i)$, also we use the notation $M(G;\alpha,\beta)$.
 \end{Definition}

\begin{Remark}
In the case $R\cong\mathbb Z_2$, $i(0)=\exp\big(\frac{J}{kT}\big)$ and $i(1)=\exp\big({-}\frac{J}{kT}\big)$ this model can be~iden\-tified with the classic isotropic Ising model~\cite{Bax}.
Therefore we will call any anisotropic or~isotropic $2$-Potts model just Ising model.
\end{Remark}

\begin{Definition}
 Consider an anisotropic $n$-Potts model $M(G;i_e)$, we~denote its partition function as $Z_n(G)$. In addition, if $n=2$, we~omit index $2$ and write just $Z(G)$.
\end{Definition}

\begin{Remark}
 For the empty graph, we~define the partition function of any $n$-Potts model to be equal to $1$, and for a disjoint set of $m$ points to be equal to $n^m$.
\end{Remark}

Now we will consider the combinatorial properties of the anisotropic $n$-Potts models (compare with~\cite[Theorem~3.2]{BEP}):

\begin{Theorem} 
Consider an anisotropic $n$-Potts model $M(G;i)$ and its partition function $Z_n(G)$.

 \begin{itemize}\itemsep=0pt
 \item Let graph $G$ be the disjoint union of graphs $G_1$ and $G_2$, then
 \begin{gather*}
 Z_n(G)=Z_n(G_1)Z_n(G_2).
 \end{gather*}

 \begin{figure}[h] \centering
 \includegraphics[scale=1]{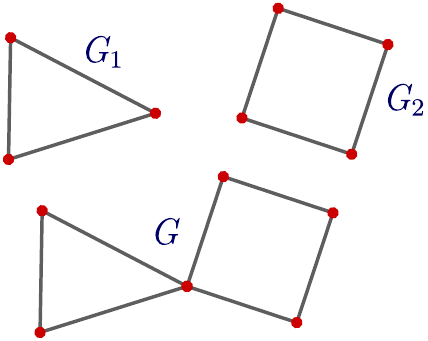}
 \caption{The joining of two graphs by the vertex $v$.}
 \end{figure}

\item Let graph $G$ be the joining of graphs $G_1$ and $G_2$ by the vertex $v$, then
\begin{gather*}
 nZ_n(G)=Z_n(G_1)Z_n(G_2).
\end{gather*}

\item Consider a graph $G$ and its edge $e$, where $e$ is neither a bridge nor a loop. Consider the graph $G/e$ obtained by contraction of $e$, and the graph $G\bs e$ obtained by deletion of $e$. Then the following formula holds
 \begin{gather*}
 Z_n(G)=(\alpha_e-\beta_e)Z_n(G/e)+\beta_e Z_n(G\bs e).
 \end{gather*}
 \end{itemize}
 \end{Theorem}

\begin{proof}\qquad

1. The statement directly follows from Definition~\ref{def:nPotts}.

2. Let us rewrite the partition function $Z_n(G)$:
\begin{gather*}
Z_n(G)=\sum\limits_{k\in\{0,\ldots,n-1\}} \sum\limits_{\sigma\colon \sigma(v)=k} W_G(\sigma).
\end{gather*}

Notice that $i(\sigma(v)-\sigma(w))=i(\sigma(v)+1-\sigma(w)-1)$, therefore for any $ i \neq j$ we have the following identity
\begin{gather*}
\sum\limits_{\sigma\colon \sigma(v)=i} W_G(\sigma)=\sum\limits_{\sigma\colon \sigma(v)=j} W_G(\sigma).
\end{gather*}
Hence we obtain
\begin{gather*}
Z_n(G)=n\sum\limits_{\sigma\colon \sigma(v)=i } W_G(\sigma) ,\qquad
\forall\, i \in \{0,\ldots, n-1\}.
\end{gather*}
Let us introduce the partial partition functions $X_k:=\!\!\sum\limits_{\sigma\colon \sigma(v)=k} \!\!W_{G_1}(\sigma)$ and $Y_k:=\!\!\sum\limits_{\sigma\colon \sigma(v)=k} \!\!W_{G_2}(\sigma)$, then we could rewrite
\begin{gather*}
Z_n(G_1)Z_n(G_2)=\bigg(\sum_k\sum_{\sigma\colon \sigma(v)=k} W_{G_1}(\sigma)\bigg)\bigg(\sum_k\sum_{\sigma\colon \sigma(v)=k} W_{G_2}(\sigma)\bigg)
\\ \hphantom{Z_n(G_1)Z_n(G_2)}
{}=(X_0+X_1+\dots +X_{n-1})(Y_0+Y_1+\dots Y_{n-1})=n^2X_0Y_0
\\ \hphantom{Z_n(G_1)Z_n(G_2)}
{}=n(X_0Y_0+X_1Y_1+\dots+X_{n-1}Y_{n-1})=n\sum_k\sum_{\sigma\colon\sigma(v)=k} W_G(\sigma)=nZ_n(G).
\end{gather*}

3. Let the edge $e$ is neither a bridge nor a loop and denote by $X$ the income in the partition function of all states such that the values of the ends of $e$ coincide and by $Y$ another part of the partition function (that of the distinct values of the ends of $e$), then
\begin{gather*}
Z_n(G)=\alpha_e X+\beta_e Y,\qquad
Z_n(G\bs e)=X+Y,\qquad
Z_n(G/e)=X
\end{gather*}
and we obtain the statement.
\end{proof}

Now let us recall the definition of the Tutte polynomial of a graph $G$.
\begin{Definition}
Let us define the Tutte polynomial $T_G(x,y)$ by the deletion-contraction recurrence relation:
\begin{enumerate}
\item If an edge $e$ is neither a bridge nor a loop, then $T_G(x, y)=T_{G\bs e}(x, y)+T_{G/e}(x, y)$.
\item If the graph $G$ consists of $i$ bridges and $j$ loops, then $T_G(x , y)=x^iy^j$.
\end{enumerate}
\end{Definition}

\begin{Theorem}[\cite{ElM,Tut}]\label{T:Tutte2}
Let $F(G)$ be a function of a graph $G$ satisfuing the following conditions:
 \begin{itemize}\itemsep=0pt
 \item $F(G)=1$, if $G$ consists of only one vertex.
 \item $F(G)=aF(G\bs e)+bF(G/e)$, if an edge $e$ is not a bridge neither a loop.
 \item $F(G)=F(G_1)F(G_2)$, if either $G=G_1\sqcup G_2$ or the intersection $G_1\cap G_2$ consists of only one vertex.
 \end{itemize}
 Then
 \begin{gather*}
 F(G)=a^{c(G)}b^{r(G)}T_G\bigg(\frac{F(K_2)}{b} ,\frac{F(L)}{a}\bigg),
 \end{gather*}
 where $K_2$ is a complete graph on two vertices, $L$ is a loop, $r(G)=v(G)-k(G)$ is a rank of $G$ and $c(G)=e(G)-r(G)$ is a corank. Here and below $e(G)$ is the number of edges in the graph~$G$.
\end{Theorem}

Now we are ready to connect the partition function $Z_n(G)$ of the isotropic $n$-Potts model $M(G;\alpha,\beta)$ with the Tutte polynomial $T_G(x,y)$ of the same graph $G$ using a well-known trick (for instance see~\cite{BEP}). Let us consider the weighted partition function
\begin{gather*}
\frac{Z_n(G)}{n^{k(G)}},
\end{gather*}
where $k(G)$ is the number of connected components in the graph $G$. It is easy to verify that the weighted partition function $\frac{Z_n(G)}{n^{k(G)}}$ satisfies Theorem~\ref{T:Tutte2}, therefore the following theorem holds:

\begin{Theorem}[Theorem 3.2~\cite{BEP}]
 \label{T:Tutte}
 The partition function $Z_n(G)$ of the $n$-Potts model $M(G;\alpha,\beta)$ coincides with the Tutte polynomial of a graph $G$ up to a multiplicative factor
\begin{gather*}
Z_n(G)=n^{k(G)}\beta^{c(G)}(\alpha-\beta)^{r(G)}T_G\bigg(\frac{\alpha+(n-1)\beta}{\alpha-\beta} ,\frac{\alpha}{\beta}\bigg).
\end{gather*}
\end{Theorem}

\begin{Example}[the bad coloring polynomial~\cite{ElM}]
Consider a graph $G$ and all possible colorings of $V(G)$ in $n$ colors. Define the bad coloring polynomial as
\begin{gather*}
B_G(n,t)=\sum\limits_{j}b_j(G,n)t^{j},
\end{gather*}
here $b_j(G,n)$ is the number of colorings such that each of them has exactly $j$ bad edges (we call an~edge ``bad'' if its ends have the same colors).
So, easy to see that $B_G(n,t)=Z_n(G)$, here~$Z_n(G)$ is the partition function of the $n$-Potts model $M(G;t,1)$. Hence, using the Theorem~\ref{T:Tutte} we immediately obtain
\begin{gather*}
B_G(n,t+1)=n^{k(G)}t^{r(G)}T_G\bigg(\frac{t+n}{t} ,t+1\bigg).
\end{gather*}
\end{Example}

\subsection[n-Potts models and the theorem of Matiyasevich]{$\boldsymbol n$-Potts models and the theorem of Matiyasevich}

The connection between the $n$-Potts models and Tutte polynomials allows us to give a simple proof of the theorem of Matiyasevich about the chromatic and flow polynomials, but at first we~introduce a few definitions.

\begin{Definition}
 A graph $A$ is called a spanning subgraph of a graph $G$, if graphs $G$ and $A$ share the same set of vertices: $V(G)=V(A)$, and the set of edges $E(A)$ is the subset of the set of edges $E(G)$.
 \end{Definition}
 \begin{Definition}
 A graph $A$ is called an edge induced subgraph (Figure~\ref{pic93}) of a graph $G$, if $A$ is induced by a subset of the set $E(G)$. Every edge induced subgraph $A$ of a graph $G$ could be completed to the spanning subgraph $A'$ by adding all the vertices of $G$ which is not contained in the subgraph $A$.
 \end{Definition}
 \begin{Definition}
 For a $n$-Potts model $M(G;i)$ we introduce the normalized partition function as follows
 \begin{gather*}
 \widetilde{Z}_n(G)=\frac{Z_n(G)}{n^{v(G)}}.
 \end{gather*}
 \end{Definition}

\begin{figure}[ht]
 \centering
 \includegraphics[scale=1]{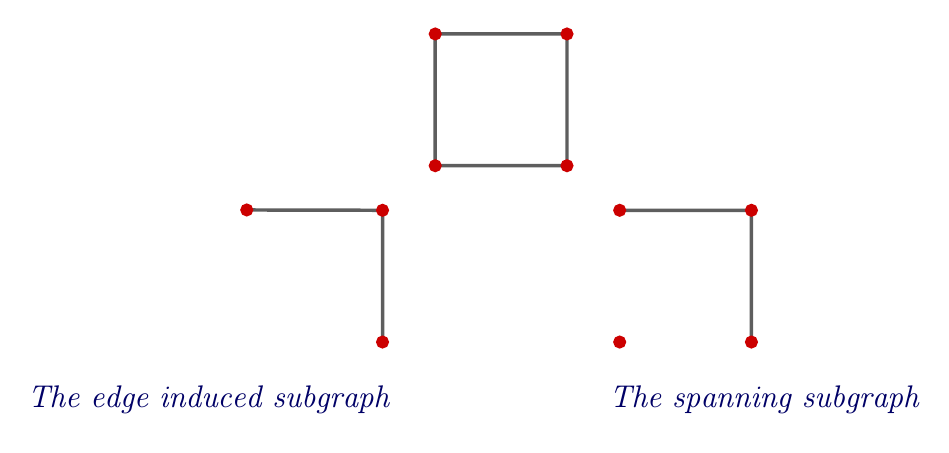}
 \caption{Edge induced and spanning subgraphs.}\label{pic93}
 \end{figure}

We start with the following lemma, which is a generalization of the high temperature formula for the Ising model:

\begin{Lemma}[Biggs formula~\cite{Big}]\label{L:Biggs}
 Let us consider two $n$-Potts models $M_1(G; i_1)$ with parame\-ters~$\alpha_1$,~$\beta_1$ and $M_2(G; i_2)$ with parameters $\alpha_2$, $\beta_2$. Then the normalized partition function $Z^1_n(G)$ of~the first model could be expressed in terms of the normalized partition functions of the models of all edge induced subgraphs of the second model:
\begin{gather*}
\widetilde{Z}^1_n(G)=q^{e(G)}\sum\limits_{A\subseteq G}\left(\frac{p}{q}\right)^{e(A)}\widetilde{Z}^2_n(A),
\end{gather*}
where $p=\frac{\alpha_1-\beta_1}{\alpha_2-\beta_2}$, and $q=\frac{\alpha_2\beta_1-\alpha_1\beta_2}{\alpha_2-\beta_2}$ $\big($we assume that $\widetilde{Z}^i_n(\varnothing)=1\big)$.
\end{Lemma}

\begin{proof}
 Let us notice that $i_1=p \cdot i_2+q$, therefore
\begin{gather*}
\widetilde{Z}^1_n(G)=\sum\limits_{\sigma\colon V(G) \rightarrow \mathbb{Z}_n}\prod_e i_1(\delta(e))=\sum\limits_{\sigma\colon V(G) \rightarrow \mathbb{Z}_n}\prod_e (pi_2(\delta(e))+q)
\\ \hphantom{\widetilde{Z}^1_n(G)}
{} =\sum\limits_{\sigma\colon V(G) \rightarrow \mathbb{Z}_n} \sum\limits_{A\subseteq G} p^{e(A)}q^{e(G)-e(A)}\prod_{e \in E(A)} i_2(\delta(e)).
\end{gather*}
 In order to complete the proof we consider the following term for a fixed $A$:
\begin{gather*}
 \sum\limits_{\sigma\colon V(G) \rightarrow \mathbb{Z}_n}\!\!\!p^{e(A)}q^{e(G)-e(A)}
 \!\!\!\!\prod_{e \in E(A)}\!\! i_2(\delta(e)) =q^{e(G)}\bigg(\frac{p}{q}\bigg)^{e(A)}\!\!\sum\limits_{\sigma\colon V(G) \rightarrow\mathbb{Z}_n}\prod_{e \in E(A)}i_2(\delta(e))
 \\ \hphantom{\sum\limits_{\sigma\colon V(G) \rightarrow \mathbb{Z}_n}\!\!\!p^{e(A)}q^{e(G)-e(A)}
 \!\!\!\!\prod_{e \in E(A)}\!\! i_2(\delta(e))}
 {}=q^{e(G)}\bigg(\frac{p}{q}\bigg)^{e(A)}n^{v(G)-v(A)}\sum\limits_{\sigma\colon V(A) \rightarrow \mathbb{Z}_n}\prod_{e \in E(A)}i_2(\delta(e))
 \\ \hphantom{\sum\limits_{\sigma\colon V(G) \rightarrow \mathbb{Z}_n}\!\!\!p^{e(A)}q^{e(G)-e(A)}
 \!\!\!\!\prod_{e \in E(A)}\!\! i_2(\delta(e))}
 {}=n^{v(G)}q^{e(G)}\sum\limits_{A\subseteq G}\bigg(\frac{p}{q}\bigg)^{e(A)}\widetilde{Z}^2_n(A).
 \tag*{\qed}
 \end{gather*}
\renewcommand{\qed}{}
\end{proof}

\begin{Proposition}
Consider two anisotropic $n$-Potts models $M_1\big(G; i^1_e\big)$ and $M_2\big(G; i^2_e\big)$.
In the same fashion we can obtain
\begin{gather}
\label{eq:aisBiggs}
\widetilde{Z}^1_n(G)=\prod_{e \in G}q_e\sum\limits_{A\subseteq G}\prod_{e \in A}\frac{p_e}{q_e}\widetilde{Z}^2_n(A),
 \end{gather}
 here $p_e=\frac{\alpha^1_e-\beta^1_e}{\alpha^2_e-\beta^2_e}$, and $q_e=\frac{\alpha^2_e\beta^1_e-\alpha^1_e\beta^2_e}{\alpha^2_e-\beta^2_e}$.
\end{Proposition}

We consider further the chromatic and flow polynomials, first of all remain some well-known definitions.
\begin{Definition} 
A coloring of the set of vertices $V(G)$ is said to be {\it proper} if the ends of each edge have different colors.
\end{Definition}

\begin{Definition}
Let $G$ be a graph with the edge set $E(G)$ and the vertex set $V(G)$, let us choose a fixed edge orientation on $G$. Then, a function $f\colon E \rightarrow \mathbb{Z}_n$ is called a nowhere-zero $n$-flow if the following conditions hold:
\begin{itemize}\itemsep=0pt
\item $\forall e \in E(G)\colon f(e) \neq 0$,
\item $\forall v \in V(G)\colon \sum \limits_{ e\in M^{+}(v)}f(e)=\sum \limits_{ e\in M^{-}(v)}f(e)$,
where $M^{+}(v)$ (respectively $M^{-}(v)$) is the set of~edges each of them is directed to (respectively from) $v$.
\end{itemize}
\end{Definition}

Next, we~formulate one of the classic results of graph theory which can be found for instance in~\cite{ElM}:

 \begin{Theorem}
The number of proper colorings of a graph $G$ in $n$ colors is the following polynomial $($called chromatic polynomial$)$ in the variable~$n$:
\begin{gather*}
\chi_G(n)=(-1)^{v(G)-k(G)}n^{k(G)}T_G(1-n, 0).
\end{gather*}
The number of nowhere-zero $n$-flows of a graph $G$ is independent on the choice of orientation and is obtained by the following polynomial $($called flow polynomial$)$ in the variable~$n$:
\begin{gather*}
C_G(n)=(-1)^{e(G)+v(G)+k(G)}T_G(0, 1-n).
\end{gather*}
 \end{Theorem}

Now we are ready to formulate and prove the theorem of Matiyasevich:
\begin{Theorem}[Matiyasevich~\cite{Mat}]\label{T:Mat}
Let us consider a graph $G$, its chromatic polynomial $\chi_G$ and its flow polynomial $C_G$, then
\begin{gather*} 
\chi_{G}(n)=\frac{(n-1)^{e(G)}}{n^{e(G)-v(G)}}\sum\limits_{A\subseteq G}\frac{C_A(n)}{(1-n)^{e(A)}},
\end{gather*}
where the summation goes through all spanning subgraphs $A$.
\end{Theorem}

\begin{proof}
Let us consider two $n$-Potts models with the special parameters: the model $M_1(G; i_1)$ with the parameters
$\alpha_1=0$, $\beta_1=1$ and the model $M_2(G; i_2)$ with the parameters $\alpha_2=1-n$, $\beta_2=1$. By Theorem~\ref{T:Tutte} we could express the partition function of the first model in terms of~the chromatic polynomial
\begin{gather*} 
\chi_G(n)=(-1)^{v(G)-k(G)}n^{k(G)}T_G(1-n, 0)=\frac{(-1)^{v(G)-k(G)-r(G)}n^{k(G)}Z^1_n(G)}{n^{k(G)}}
\\ \hphantom{\chi_G(n)}
=(-1)^{v(G)-k(G)-r(G)}n^{v(G)}\widetilde{Z}^1_n(G)=n^{v(G)}\widetilde{Z}^1_n(G).
\end{gather*}
So we have
 \begin{gather*}
 \widetilde{Z}^1_n(G)=\frac{\chi_G(n)}{n^{v(G)}}.
 \end{gather*}
{\samepage
Analogously, we~express the partition function of the second model in terms of the flow poly\-nomial
 \begin{gather}
 C_G(n)=(-1)^{e(G)+v(G)+k(G)}T_G(0,1-n)=
 \frac{(-1)^{e(G)+v(G)+k(G)-r(G)}Z^2_n(G)}{n^{k(G)}n^{v(G)-k(G)}}\nn
 \\ \hphantom{C_G(n)}
 {}=(-1)^{e(G)}\widetilde {Z}^2_n(G).\label{eq:flow}
\end{gather}

}\noindent
So we have
\begin{gather} \label{eq:collor}
 \widetilde{Z}^2_n(G)=(-1)^{e(G)}C_G(n).
\end{gather}
Then by Lemma~\ref{L:Biggs} after the substitutions~\eqref{eq:flow} and~\eqref{eq:collor} we obtain
\begin{gather*} 
 \frac{\chi_G(n)}{n^{v(G)}}=\frac{(n-1)^{e(G)}}{n^{e(G)}}\sum_{A'\subseteq G}\frac{C_{A'}(n)}{(1-n)^{e(A')}},
\end{gather*}
 where the summation goes through all edge induced subgraphs $A'$.

We finish the proof by noticing that the edge induced subgraph differs from the spanning subgraph by the set of isolated vertices. Therefore we can complete each edge induced subgraph to its corresponding spanning subgraph and then replace the summation over all edge induced subgraph by the summation over all spanning subgraph, because the value of the each flow polynomial $C_{A'}$ remains the same and finally we obtain
\begin{gather*}
 \chi_G(n)=\frac{(n-1)^{e(G)}}{n^{e(G)-v(G)}}\sum_{A\subseteq G}\frac{C_{A}(n)}{(1-n)^{e(A)}}. \tag*{\qed}
\end{gather*}
\renewcommand{\qed}{}
\end{proof}

 Note that we could produce series of statements that look like Theorem~\ref{T:Mat}:

\begin{Theorem}
Let us consider a graph $G$, then we can obtain the following formulas
\begin{gather}
\label{eq:genm}
 n^{k(G)}\beta_1^{c(G)}(\alpha_1-\beta_1)^{r(G)}T_G\bigg(\frac{\alpha_1+(n-1)\beta_1}{\alpha_1-\beta_1} ,\frac{\alpha_1}{\beta_1}\bigg)=q^{e(G)}\sum_{A\subseteq G}\bigg(\frac{p}{q}\bigg)^{e(A)}\chi_A(n),
\end{gather}
where $p=-\alpha_1+\beta_1$, $q=\alpha_1$, and the summation (here and below) goes through all spanning subgraphs $A$,
\begin{gather}
\label{eq:invers}
 C_{G}(n)=(n-1)^{e(G)}\sum_{A\subseteq G}\frac{n^{e(A)-v(G)}}{(1-n)^{e(A)}}\chi_{A}(n),
\\
 n^{k(G)-v(G)}\beta_1^{c(G)}(\alpha_1-\beta_1)^{r(G)}T_G
 \bigg(\frac{\alpha_1+(n-1)\beta_1}{\alpha_1-\beta_1} ,\frac{\alpha_1}{\beta_1}\bigg)\nn
 \\ \qquad
{} =q_1^{e(G)}\sum_{A}\bigg(\frac{p_1}{q_1}\bigg)^{e(A)}(-1)^{e(A)}C_{A}(n),\label{eq:t-flow}
\end{gather}
where $p_1=\frac{\beta_1-\alpha_1}{n}$, $q_1=\frac{\alpha_1-(1-n)\beta_1}{n}$,
\begin{gather}
 (-1)^{e(G)}C_G(n)\nn
 \\ \qquad
 =q_2^{e(G)}\sum_{A}\bigg(\frac{p_2}{q_2}\bigg)^{e(A)}
 n^{k(A)-v(G)}\beta_1^{c(A)}(\alpha_1-\beta_1)^{r(A)}T_A
 \bigg(\frac{\alpha_1+(n-1)\beta_1}{\alpha_1-\beta_1} ,\frac{\alpha_1}{\beta_1}\bigg),
\label{eq:flow-t}
\end{gather}
where $p_2=\frac{n}{\beta_1-\alpha_1}$ and $q_2=\frac{\alpha_1-(1-n)\beta_1}{\alpha_1-\beta_1}$.
\end{Theorem}
\begin{proof}
Let us consider two $n$-Potts models:
\begin{itemize}\itemsep=0pt
 \item Models $M_1(G; \alpha_1,\beta_1)$ and $M_2(G; 0, 1)$ for the proof of the formula~\eqref{eq:genm}.
 \item The specification of the first case: $M_1(G; 1-n, 1)$ and the same $M_2(G; 0,1)$ for the proof of the formula~\eqref{eq:invers}.
 \item Models $M_1(G; 1-n,1)$ and $M_2(G; \alpha_1, \beta_1)$ with the parameters $\alpha_1$, $\beta_1$ for the proof of the formula~\eqref{eq:t-flow}.
 \item And finally, models $M_1(G; \alpha_1, \beta_1)$ and $M_2(G; 1-n,1)$ for the proof of formula~\eqref{eq:flow-t}.
\end{itemize}
Now it is left to repeat step by step the proof of Theorem~\ref{T:Mat} for these two models.
\end{proof}

\begin{Remark}
We notice that the formula~\eqref{eq:invers} naturally can be considered as ``inversion'' of~Theorem~\ref{T:Mat}.
\end{Remark}

\subsection{Shifting the order in the Potts models}

Biggs Lemma~\ref{L:Biggs} allows us to relate the values of the partition functions of the $n$-Potts models with fixed $n$, but different values of parameters $\alpha$ and $\beta$. The goal of the current subsection is to present a method for connecting partition functions of the $n$-Potts models for different $n$. We~will call it \textit{shifting order formulas}.

The first method is based on the multiplicativity property of the \textit{complete} flow polynomial.
\begin{Definition}
Let $G$ be a graph with the edge set $E(G)$ and the vertex set $V(G)$, let us chose a fixed edge orientation on $G$. Then, a function $f\colon E \rightarrow \mathbb{Z}_n$ is called an $n$-flow if the following condition holds
 \begin{gather*}
 \forall v \in V(G)\colon\ \sum \limits_{ e\in M^{+}(v)}f(e)=\sum \limits_{ e\in M^{-}(v)}f(e),
 \end{gather*}
 here again $M^{+}(v)$ (respectively $M^{-}(v)$) is the set of edges each of them is directed to (respectively from) $v$.
\end{Definition}

Let us formulate a few well known results concerning a flow polynomial and a number of all $n$-flows. The proofs could be found for example in~\cite{Sok}.

\begin{Proposition}\label{Prop:fc-sumflow}
Denote the number of all $n$-flows on a graph $G$ by ${FC}_G(n)$, then ${FC}_G(n)$ is independent of the choice of an orientation and the following identity holds
\begin{gather*}
{FC}_G(n)=\sum\limits_{A\subseteq G} C_A(n),
\end{gather*}
where the summation goes through all spanning subgraphs $A$ of the graph $G$.
\end{Proposition}

\begin{Proposition}
The number of all $n$-flows on a graph is the following polynomial $($called complete flow polynomial$)$
\begin{gather*}
{FC}_G(n)=n^{e(G)-v(G)+k(G)},
\end{gather*}
where $e(G)$, $v(G)$, $k(G)$ are numbers of edges, vertices and connected components in the graph~$G$ correspondingly.
\end{Proposition}

\begin{Proposition} \label{prop:c-sumflow}
The flow polynomial $C_G(n)$ of a graph $G$ could be expressed in terms of the complete flow polynomials of its spanning subgraphs by the following identity:
\begin{gather*}
C_G(n)=\sum _{A\subseteq G} (-1)^{e(G)-e(A)}{FC}_A(n).
\end{gather*}
\end{Proposition}

The complete flow polynomial ${FC}_G(n)$ is a multiplicative invariant: ${FC}_G(n_1n_2)={FC}_G(n_1)\allowbreak \times {FC}_G(n_2)$, therefore we are ready to formulate the following theorem:

\begin{Theorem} \label{T:sh}
The partition function $Z_{n_1n_2}(G)$ of the $n_1n_2$-Potts model $M(G; \alpha_1, \beta_1)$ could be expressed in terms of the partition functions $Z_{n_1}(A)$ and $Z_{n_2}(A)$ of the $n_1$-Potts model $M_1(A; \alpha_1, \beta_1)$ and $n_2$-Potts model $M_2(A; \alpha_1, \beta_1)$ of all spanning subgraphs $A$ of the graph $G$ correspondingly.
 \end{Theorem}

\begin{proof}Indeed, by Theorem~\ref{T:Tutte} and the formula~\eqref{eq:t-flow} we have
\begin{gather*}
Z_{n_1n_2}(G)=\gamma_G T_G \bigg(\frac{\alpha_1+(n_1n_2-1)\beta_1}{\alpha_1-\beta_1} ,\frac{\alpha_1}{\beta_1}\bigg)=\sum_{A\subseteq G} \lambda_A C_A(n_1 n_2).
\end{gather*}

From Proposition~\ref{prop:c-sumflow} we obtain
\begin{gather*}
\sum_{A\subseteq G} \lambda_A C_A(n_1n_2)=\sum_{A\subseteq G} \lambda_A \bigg(\sum_{A'\subseteq A} (-1)^{e(A)-e(A')} {FC}_{A'}(n_1 n_2)\bigg)
\\ \hphantom{\sum_{A\subseteq G} \lambda_A C_A(n_1n_2)}
{}=\sum_{A\subseteq G} \omega_A {FC}_{A}(n_1 n_2)=\sum_{A\subseteq G} \omega_A {FC}_{A}(n_1){FC}_{A}(n_2),
\end{gather*}
notice that we used for the second resummations the following simple observation: if $X$ is a~span\-ning subgraph of $Y$, which is a spanning subgraph of graph $Z$, so $X$ is a spanning subgraph of~a~graph $Z$. We~omit this remark below.

The Proposition~\ref{Prop:fc-sumflow} implies
\begin{gather*}
\sum_{A\subseteq G} \omega_A {FC}_{A}(n_1){FC}_{A}(n_2)=\sum_{A\subseteq G} \omega_A
\bigg(\sum_{A'\subseteq A}C_{A'}(n_1)\bigg)\bigg(\sum_{A''\subseteq A}C_{A''}(n_2)\bigg)
\\ \hphantom{\sum_{A\subseteq G} \omega_A {FC}_{A}(n_1){FC}_{A}(n_2)}
{}=\sum_{A'\subseteq G}\sum_{A'' \subseteq G} \mu_{A'A''}C_{A'}(n_1)C_{A''}(n_2).
\end{gather*}

Finally, with the help of formula~\eqref{eq:flow-t} and Theorem~\ref{T:Tutte} we obtain
\begin{gather*}
\sum_{A'\subseteq G}\sum_{A'' \subseteq G} \mu_{A'A''}\bigg(\sum_{B\subseteq A'} \delta_B Z_{n_1}(B)\bigg)\bigg(\sum_{C\subseteq A''} \delta_C Z_{n_2}(C)\bigg)=\sum_{A'\subseteq G}\sum_{A'' \subseteq G} \eta_{A'A''} Z_{n_1}(A') Z_{n_2}(A''),
\end{gather*}
where $\eta_{A'A''}$ are some constants, appeared after the resummations.
\end{proof}

\begin{Remark}[convolution formula~\cite{Koo}]
It seems extremely interesting and fruitful to compare Lemma~\ref{L:Biggs} and Theorem~\ref{T:sh} with the convolution formula
\begin{gather*}
T_G(x, y) = \sum_{A\subseteq E(G)}T_{G|A}(0,y)T_{G/A}(x,0),
\end{gather*}
here the summation is over all possible subsets of $E(G)$, here $G|A$ is a graph obtained by the restriction of $G$ on the edge subset $A$ and $G/A$ is a graph obtained from $G$ by the contraction of all edges from $A$ (see~\cite{ElM} for more details).
\end{Remark}

Our second method is based on the Tutte identity for the chromatic polynomial:

\begin{Theorem}[\cite{ElM}] \label{T:Tid}
 Consider a graph $G$ with the set of edge $V(G)$ then the following formula holds
 \begin{gather*}
 \chi_G(n_1+n_2)=\sum_{B \subseteq V(G)} \chi_{G|B}(n_1)\chi_{G|B^c}(n_2),
 \end{gather*}
 where $G|B$ $(G|B^c)$ is the restriction of $G$ on the vertex subset $B \in V(G)$ $(B^c \in V(G)$, where $B^c=V(G)\setminus B)$.
 \end{Theorem}

 Using this fact we can formulate the following theorem:

\begin{Theorem} 
The partition function $Z_{n_1+n_2}(G)$ of the $n_1+n_2$-Potts model $M(G; \alpha_1, \beta_1)$ could be expressed in terms of the partition functions $Z_{n_1}(A)$ and $Z_{n_2}(A)$ of the $n_1$-Potts model $M_1(A; \alpha_1, \beta_1)$ and $n_2$-Potts model $M_2(A; \alpha_1, \beta_1)$ of all spanning subgraphs $A$ of the graph $G$ correspondingly.
 \end{Theorem}

 \begin{proof}
The proof is very similar to the proof of Theorem~\ref{T:sh}. Again, from Theorem~\ref{T:Tutte} and the formula~\eqref{eq:genm} we have
\begin{gather*}
Z_{n_1+n_2}(G)=\gamma_G T_G \bigg(\frac{\alpha_1+(n_1+n_2-1)\beta_1}{\alpha_1-\beta_1} ,\frac{\alpha_1}{\beta_1}\bigg)=\sum_{A\subseteq G} \lambda_A \chi_A(n_1+n_2).
\end{gather*}
From Theorem~\ref{T:Tid} we obtain
\begin{gather*}
\sum_{A\subseteq G} \lambda_A \chi_A(n_1+n_2)=\sum_{A \subseteq G}\lambda_A\bigg(\sum_{B \subseteq V(A)} \chi_{A|B}(n_1)\chi_{A|B^c}(n_2)\bigg)
\\ \hphantom{\sum_{A\subseteq G} \lambda_A \chi_A(n_1+n_2)}
{}=\sum_{A \subseteq G}\lambda_A\Bigg(\sum_{B \subseteq V(A)}\!\! \bigg(\sum_{A_1 \subseteq A|B}\!\! \omega_{A_1}Z_{n_1}(A_1)\bigg)\bigg(\sum_{A_2 \subseteq G|B^c} \!\!\omega_{A_2}Z_{n_2}(A_2)\bigg)\Bigg)=
\\ \hphantom{\sum_{A\subseteq G} \lambda_A \chi_A(n_1+n_2)}
{}=\sum_{A \subseteq G} \sum_{B \subseteq V(A)} \sum_{A_1 \subseteq A|B} \sum_{A_2 \subseteq A|B^c} \mu_{A_1 A_2}Z_{n_1}(A_1)Z_{n_2}(A_2).
\end{gather*}
Let us complete each subgraph $A_1$ (each $A_2$) to the corresponding spanning subgraph of $G$ by adding isolating vertices
\begin{gather*}
\sum_{A \subseteq G} \sum_{B \subseteq V(A)} \sum_{A_1 \subseteq A|B} \sum_{A_2 \subseteq A|B^c} \mu_{A_1 A_2}Z_{n_1}(A_1)Z_{n_2}(A_2)=\sum_{A'\subseteq G}\sum_{A'' \subseteq G} \eta_{A'A''} Z_{n_1}(A') Z_{n_2}(A''),
\end{gather*}
where $\eta_{A'A''}$ are again some constants, appeared after the resummations.
\end{proof}

\section{Star-triangle equation for Ising and Potts models}\label{sec:ST}

\subsection{General properties} Let us rewrite the partition function of the anisotropic $n$-Potts models in the so called Fortuin--Kasteleyn representation:
 \begin{Proposition}[compare with the formula $(2.7)$ from~\cite{Sok}]\label{d1}
 Consider the anisotropic $n$-Potts model $M(G;i_e)$, then its partition function could be expressed as follows
 \begin{gather}\label{statsumma}
Z_n(G)=\sum_{\sigma}\prod_{e\in E}(\beta_{e}+(\alpha_{e}-\beta_{e})\delta(\sigma_{e}))=\prod_{e\in E} \beta_{e}\sum_{\sigma}\prod_{e\in E}(1+(t_{e}-1)\delta(\sigma_{e})),
\end{gather}
where $\delta(\sigma_{e})$ is a value of standard Kronecker delta function of the values of $\sigma$ on the boundary vertices of the edge $e$ and $t_{e}=\frac{\alpha_{e}}{\beta_{e}}$ is a reduced weight of the edge~$e$.
\end{Proposition}
 \begin{proof}
 Indeed, it is easy to see that if $\sigma(v)=\sigma(w)$:
 \begin{gather*}
 i_e(\delta(e))=i_e(\sigma(v)-\sigma(w))=\alpha_e
 =\beta_{e}+(\alpha_{e}-\beta_{e})\delta(\sigma_{e})
 =\beta_{e}+(\alpha_{e}-\beta_{e})\delta(\sigma(v),\sigma(w)),
 \end{gather*}
 and if $\sigma(v) \neq \sigma(w)$:
 \begin{gather*}
 i_e(\delta(e))=i_e(\sigma(v)-\sigma(w))=\beta_e
 =\beta_{e}+(\alpha_{e}-\beta_{e})\delta(\sigma_{e})
 =\beta_{e}+(\alpha_{e}-\beta_{e})\delta(\sigma(v),\sigma(w)).\!\!\!\!
 \tag*{\qed}
 \end{gather*}
\renewcommand{\qed}{}
\end{proof}

Also, we~introduce the \textit{boundary} partition function of the $n$-Potts models:
\begin{Definition} \label{d2}
 Let $G$ be a graph (with possible loops and multiple edges) with the set of vertices $V$, the set of edges $E$ and the \textit{boundary} subset $S\subseteq V $ of enumerated vertices: $S=\{v_1, v_2, \dots , v_k\}$. The \textit{boundary} partition function on $G$ is defined by the following expression
\begin{gather*}
Z_{n;S(\textbf{A})}(G)=\sum_{\sigma_{\textbf{A}}}\prod_{e\in E}(\beta_{e}+(\alpha_{e}-\beta_{e})\delta(\sigma_{e})),
\end{gather*}
 where $\textbf{A}=\{a_1, a_2, \dots, a_k\}$, $\forall i\colon a_i \in \mathbb Z_n$ is the set of fixed values, and the summation is over such states $\sigma_{\textbf{A}}$ that $\sigma_{\textbf{A}}(v_i)=a_i$.
\end{Definition}
\begin{Remark}
If $n=2$, we~will omit the index $2$ and will write just $Z_{S(\textbf{A})}(G)$.
\end{Remark}

The next Lemma connects boundary and ordinary partition functions:

 \begin{Lemma} \label{lemma1}
Consider two graphs $G_1=(V_1, E_1)$ and $G_2=(V_2, E_2)$ with the only common vertices in the \textit{boundary} subset $S=\{v_1, v_2, \dots, v_n\}$ in $V_1$ and $V_2$. We~can glue these graphs and obtain the third graph $G=(V,E)$, where $E=E_1\sqcup E_2$, $V=V_1\cup_{S} V_2$. Then, the following identity holds
\begin{gather*}
Z_n(G)=\sum_{\textbf{A}}Z_{n;S(\textbf{A})}(G_1)Z_{n;S(\textbf{A})}(G_2),
\end{gather*}
where the summation is over all possible sets $\textbf{A}$.
\end{Lemma}

\begin{figure}[h] \centering
 \includegraphics[scale=1]{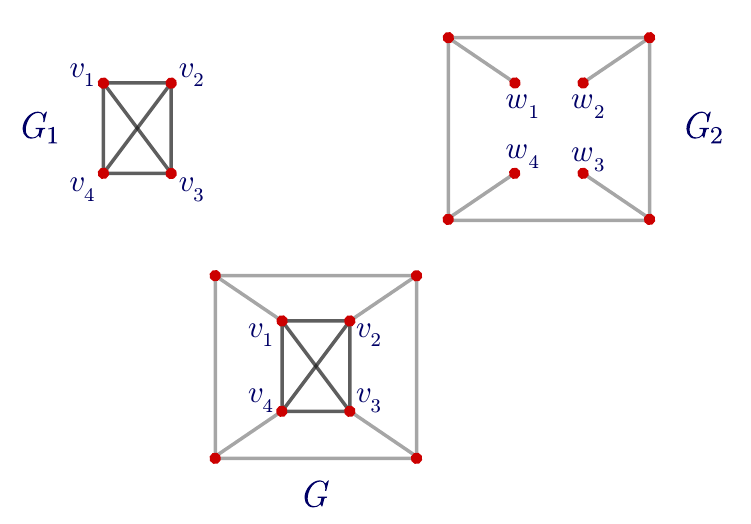}
 \caption{$G$ is obtained by merging of $S=\{v_1, v_2, v_3, v_4\}$.}
\end{figure}

 \begin{proof} The formula is obtained directly from the Proposition~\ref{d1} and Definitions~\ref{d2}. \mbox{Indeed}, by the Definition~\ref{d2} we can write down $Z_n(G)=\sum\limits_{\textbf{A}}Z_{n;S(\textbf{A})}(G), $ but also $Z_{n;S(\textbf{A})} (G)=Z_{n;S_1(\textbf{A})} (G_1)Z_{n;S_2(\textbf{A})} (G_2)$.
 \end{proof}

\begin{Remark}
The latter property of a partition function (Lemma~\ref{lemma1}) allow us to consider $n$-Potts model partition function as a discrete version of the topological quantum field theory in the Atiyah formalism~\cite{Atiah}, where
\begin{gather*}
 TQFT\colon\ {\rm Cob}\to {\rm Vect}
 \end{gather*}
 is a functor from the category of cobordisms to the category of vector spaces.
 \end{Remark}

\subsection[The case n=2]{The case $\boldsymbol{n=2}$}

In this subsection we consider the case $n=2$. Our first goal is to find such conditions that the partition function~\eqref{statsumma} is invariant under the star-triangle transformation which changes the subgraph $\Omega$ to the subgraph $\Omega'$.
We derive these conditions with the use of the boundary partition functions: consider a~graph~$G$ with the subgraph $\Omega$, then using Lemma~\ref{lemma1} for graphs~$\Omega$ and $G-\Omega$ we obtain the following identity
\begin{gather*}
 Z(G)=\sum_{\textbf{A}} Z_{S(\textbf{A})}(\Omega) Z_{S(\textbf{A})} (G - \Omega),
\end{gather*}
 where $S=\{v_1, v_2, v_3\}$ (Figure~\ref{fig startr}).
 After the star-triangle transformation, we~obtain a graph $G'$ with the following partition function
 \begin{gather*}
 Z(G')=\sum_{\textbf{A}} Z_{S(\textbf{A})}(\Omega') Z_{S(\textbf{A})} (G' - \Omega').
 \end{gather*}
\begin{figure}[h] \centering
\includegraphics[scale=1]{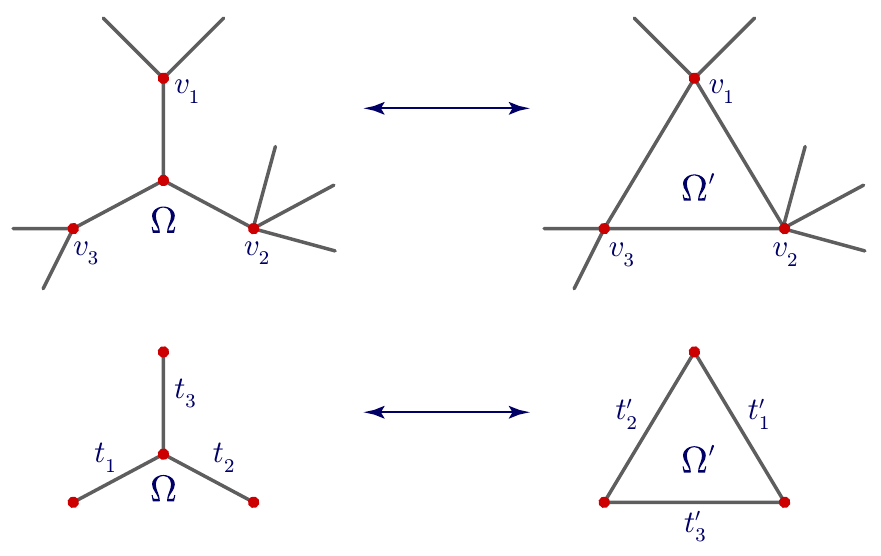}
\caption {Star-triangle transformation.}\label{fig startr}
\end{figure}

 \noindent
 Due to the fact that the star-triangle transformation does not change edges of the graph $G-\Omega$, we~deduce that
 $\forall \textbf{A}\colon Z_{S(\textbf{A})} (G - \Omega)=Z_{S(\textbf{A})} (G' - \Omega')$.
 Therefore, the sufficient and necessary conditions for the invariance of the partition function are the following
 \begin{gather}
 \label{cond}
\forall \textbf{A}\colon\ Z_{S(\textbf{A})} (\Omega)=Z_{S(\textbf{A})} (\Omega').
\end{gather}
We write them down in detail. Let us note that these conditions do not depend on the states of the vertices (see Figure~\ref{fig startr2}), but depend on the number and the
positions of the vertices with equal states. Therefore, we~have the following possibilities:
\begin{figure}[h] \centering
\includegraphics[scale=1 ]{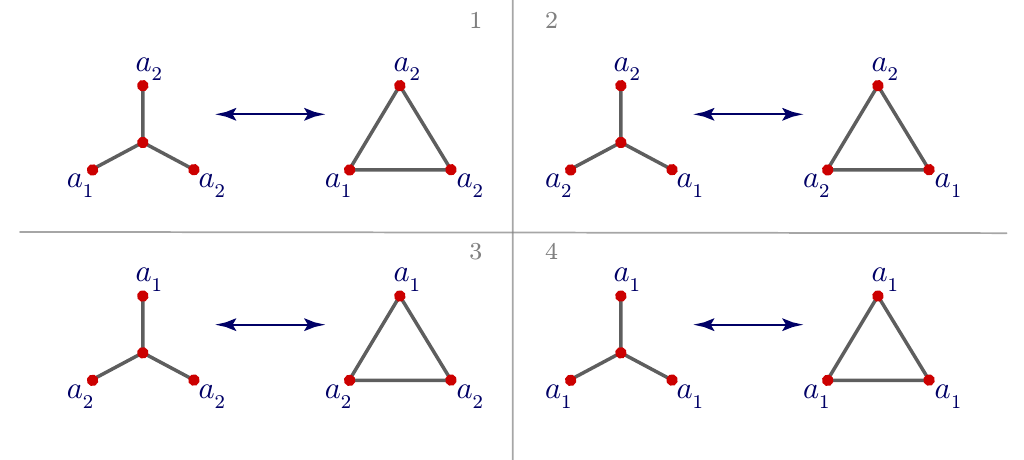}
\caption{Different possibilities.}\label{fig startr2}
\end{figure}

\begin{itemize}\itemsep=0pt
	\item two states in the triangle are the same, then the central vertex either has the same state or has the different state, then $\alpha_1\beta_2\beta_3+\beta_1\alpha_2\alpha_3\mapsto \alpha'_1\beta'_2\beta'_3$ and two more maps after permuting indexes,
	\item all states are the same, then $\alpha_1\alpha_2\alpha_3+\beta_1\beta_2\beta_3\mapsto \alpha'_1\alpha'_2\alpha'_3$.
\end{itemize}

In this way we obtain the following equations
\begin{gather} \label{eq1}
\begin{cases}
\alpha_1\beta_2\beta_3+\beta_1\alpha_2\alpha_3= \alpha'_1\beta'_2\beta'_3,\\
\alpha_2\beta_1\beta_3+\beta_2\alpha_1\alpha_3= \alpha'_2\beta'_1\beta'_3,\\
\alpha_3\beta_1\beta_2+\beta_3\alpha_1\alpha_2= \alpha'_3\beta'_1\beta'_2,\\
\alpha_1\alpha_2\alpha_3+\beta_1\beta_2\beta_3= \alpha'_1\alpha'_2\alpha'_3.
\end{cases}
\end{gather}
After the substitution
\begin{gather*}
t_i = \frac{\alpha_i}{\beta_i}
\end{gather*}
we rewrite~\eqref{eq1} as
\begin{gather*}
\begin{cases}
\beta_1\beta_2\beta_3(t_1+t_2t_3)=\beta'_1\beta'_2\beta'_3t'_1,\\
\beta_1\beta_2\beta_3(t_2+t_1t_3)=\beta'_1\beta'_2\beta'_3t'_2,\\
\beta_1\beta_2\beta_3(t_3+t_1t_2)=\beta'_1\beta'_2\beta'_3t'_3,\\
\beta_1\beta_2\beta_3(t_1t_2t_3+1)=\beta'_1\beta'_2\beta'_3t'_1t'_2t'_3.
\end{cases}
\end{gather*}

This set of equations defines a correspondence which preserves the Ising model partition function if we mutate the graph $G$ to $G'$. Let us denote the product $\beta_1\beta_2\beta_3$ by $\beta$ and the product $\beta'_1\beta'_2\beta'_3$ by $\beta'$. Then, we~obtain the following map from the $(t,\beta)$-variables to the $(t',\beta')$ variables, we~will call it $\widetilde{F}$,
\begin{gather}
\widetilde{F}(t_1, t_2, t_3,\beta)=(t'_1, t'_2, t'_3,\beta')\colon\nn\\
t'_1 =\sqrt{\dfrac{(t_1+t_2t_3)(t_1t_2t_3+1)}{(t_2+t_1t_3)(t_3+t_1t_2)}},\nn\\
t'_2 = \sqrt{\dfrac{(t_2+t_1t_3)(t_1t_2t_3+1)}{(t_1+t_2t_3)(t_3+t_1t_2)}},\nn\\
t'_3 = \sqrt{\dfrac{(t_3+t_1t_2)(t_1t_2t_3+1)}{(t_1+t_2t_3)(t_2+t_1t_3)}},\nn\\
\beta'=\beta\sqrt{\frac{(t_1+t_2t_3)(t_3+t_1t_2)(t_2+t_1t_3)}{(t_1t_2t_3+1)}}.\label{Tchange}
\end{gather}
\begin{Remark}
Formally speaking to define a map on the space of edge weight adopted to the star-triangle transformation we have to resolve the map $\widetilde{F}$ somehow for the parameters $\beta_i$. For~example one can take the following one
\[
\beta'_i=\beta_i(\beta'/\beta)^{1/3}.
\]
Actually, the choice of a resolution is not important in what follows.
\end{Remark}

\begin{Remark}
We choose the positive branch of the root function for real positive values of~vari\-ables $t_i$ for purposes emphasized further. This is relevant to the almost positive version of the orthogonal grassmanian. See~\cite{GalPil} for more details about the connection of the Ising model and positive orthogonal grassmanian.
\end{Remark}

\subsection[The case n>2]{The case $\boldsymbol{n \neq 2}$}

Let us demonstrate how the method, described above, works for the star-triangle transformation in the case $n \geq 3$. Using the same ideas as in the previous subsection we could obtain the following conditions
\begin{equation}
\begin{cases}
\beta_1\beta_2\beta_3(t_1+t_2t_3+n-2)= \beta'_1\beta'_2\beta'_3t'_1,\\[.5ex]
\beta_1\beta_2\beta_3(t_2+t_1t_3+n-2)= \beta'_1\beta'_2\beta'_3t'_2,\\[.5ex]
\beta_1\beta_2\beta_3(t_3+t_1t_2+n-2)= \beta'_1\beta'_2\beta'_3t'_3,\\[.5ex]
\beta_1\beta_2\beta_3(t_1t_2t_3+n-1)= \beta'_1\beta'_2\beta'_3t'_1t'_2t'_3,\\[.5ex]
\beta_1\beta_2\beta_3(t_1+t_2+t_3+n-3)= \beta'_1\beta'_2\beta'_3.
\end{cases}\label{eq:q>3}
\end{equation}
Here the last equation follows from the extra case in which all states are different.
\begin{figure}[h]
 \centering
\includegraphics[scale=1 ]{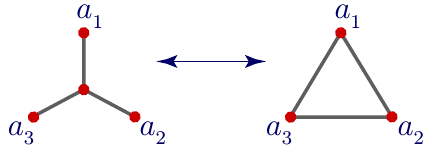}
\caption{The extra case.}
\end{figure}

In general, the system~\eqref{eq:q>3} does not have a solution and the star-triangle transformation is~not possible. But, if $t_i$ satisfy the special condition, partition function of the $n$-Potts model is still invariant under the star-triangle transformation.
\begin{Proposition}
The system~\eqref{eq:q>3} together with equation
\begin{gather} \label{condeq}
t_1t_2t_3=t_1t_2+t_2t_3+t_3t_1+(n-1)(t_1+t_2+t_3)+n^2-3n+1
\end{gather}
has a solution in terms of prime variables.
\end{Proposition}
\begin{proof}
Using the first three and the last equations of~\eqref{eq:q>3} we immediately obtain the expressions for $t'_i$ and $\frac{\beta'_1\beta'_2\beta'_3}{\beta_1\beta_2\beta_3}$:
\begin{gather*}
\begin{cases}
\dfrac{\beta'_1\beta'_2\beta'_3}{\beta_1\beta_2\beta_3}=t_1+t_2+t_3+n-3,
\\[2ex]
t'_1=\dfrac{t_1+t_2t_3+n-2}{t_1+t_2+t_3+n-3},
\\[2ex]
t'_2=\dfrac{t_2+t_1t_3+n-2}{t_1+t_2+t_3+n-3},
\\[2ex]
t'_3=\dfrac{t_3+t_1t_2+n-2}{t_1+t_2+t_3+n-3}.
\end{cases}
\end{gather*}
Substitute these expressions into the fourth equation of~\eqref{eq:q>3} and obtain the equation
\begin{gather} \label{fineq}
t_1t_2t_3+n-1=\frac{(t_1+t_2t_3+n-2)(t_2+t_1t_3+n-2)(t_3+t_1t_2+n-2)}{(t_1+t_2+t_3+n-3)^2}.
\end{gather}

By the straightforward computation, we~retrieve that the identity~\eqref{condeq} is the consequence from the equation~\eqref{fineq}.
\end{proof}

\begin{Corollary}Partition function of the $n$-Potts model $(n\geq 3)$ is invariant under the star-triangle transformation if and only if the system~\eqref{eq:q>3} with the equation~\eqref{condeq} hold.
\end{Corollary}
Below we present two nontrivial specialization of the partition function of $n$-Potts model which are agreed with the system~\eqref{eq:q>3},~\eqref{condeq}.

\begin{Example} Consider a graph $G$ and equip each $e\in E(G)$ with sign $+$ or $-$. Let us consider the $n$-Potts model $M_k(G, \alpha_e, \beta_e)$ with following parameters:
\begin{itemize}\itemsep=0pt
 \item for all $e \in E$ equipped with $+$ the parameters $\alpha_e$, $\beta_e$ equal $\alpha_e=A_{+}=-t^{-\frac{3}{4}}$, $\beta_e=B_{+}=t^{\frac{1}{4}}$,
 \item for all $e \in E $ equipped with $-$ the parameters $\alpha_e$, $\beta_e$ equal $\alpha_e=A_{-}=-t^{\frac{3}{4}}$, $\beta_e=B_{-}=t^{-\frac{1}{4}}$,
 \item and $n=t+\frac{1}{t}+2$ (we suppose that parameter $t$ is chosen such that $n \in \mathbb{N}$).
\end{itemize}
Let the graph $G$ has a triangle subgraph, the edges of which have signs $+$, $-$, $-$. The reduced weights of edges are $t_1=\frac{A_{+}}{B_{+}}=-\frac{1}{t}$, $t_2=\frac{A_{-}}{B_{-}}=-t$, $t_3=\frac{A_{-}}{B_{-}}=-t$. It is easy to see that these $t_i$ satisfy the equation~\eqref{condeq}.

We notice that the signed graph $G$ could be considered as the signed Tait graph for a dia\-gram~$D(K)$ of a knot $K$ (\cite{Wu1}, the chapter ``Knot invariants from edge-interaction models"). Moreover, the value of the Jones polynomial of the knot $K$ at the point~$n$ is closely related with the partition function of the $n$-Potts model $M_k(G, \alpha_e, \beta_e)$ (see~\cite[equation~(7.17)]{Wu1}). Thus, the identification of the third Reidemeister move of the diagram $D(K)$ with the star-triangle transformation of the signed graph $G$ is agreed with the star-triangle transformation defined by~the system~\eqref{eq:q>3},~\eqref{condeq} for the $n$-Potts model $M_k(G, \alpha_e, \beta_e)$.
\end{Example}

Our second example is about the models of bond percolation. Firstly, we~briefly give their definitions:
\begin{Definition}[bond percolation~\cite{Gri}]
 Consider a graph $G$. An edge $e \in E(G)$ is considered to be open with probability $p_e$ or closed with probability $1-p_e$. We suppose that all edges might be closed or open independently. One is interested in probabilistic properties of cluster formation (i.e. maximal connected sets of closed edges of the graph~$G$).
\end{Definition}

\begin{Example}The bond percolation models could be considered as a limit $n \to 1$ of the $n$-Potts models at the level of the boundary partition functions~\cite{Car}. This identification corresponds to the specialization of the system~\eqref{eq:q>3},~\eqref{condeq} by $n\to 1$.

Substitute $t_i=\frac{1}{p_i}$, $t'_i=\frac{1}{p'_i}$ and $n=1$ in~\eqref{eq:q>3} and~\eqref{condeq}, then
\begin{gather*}
\begin{cases}
(p_1+p_2p_3-p_1p_2p_3)\alpha_1\alpha_2\alpha_3=p'_2p'_3\alpha'_1\alpha'_2\alpha'_3,
\\
(p_2+p_1p_3-p_1p_2p_3)\alpha_1\alpha_2\alpha_3=p'_1p'_3\alpha'_1\alpha'_2\alpha'_3,
\\
(p_3+p_1p_2-p_1p_2p_3)\alpha_1\alpha_2\alpha_3=p'_1p'_2\alpha'_1\alpha'_2\alpha'_3,
\\
\alpha_1\alpha_2\alpha_3=\alpha'_1\alpha'_2\alpha'_3,
\\[1ex]
\dfrac{1}{p_1p_2}+\dfrac{1}{p_2p_3}+\dfrac{1}{p_1p_3}-1=\dfrac{1}{p_1p_2p_3}.
\end{cases}
\end{gather*}
After simplifications we obtain the condition for the star-triangle transformation of the bond percolation models (for instance, see~\cite{Gri})
\begin{gather*}
\begin{cases}
p_1+p_2p_3-p_1p_2p_3=p'_2p'_3,\\
p_2+p_1p_3-p_1p_2p_3=p'_1p'_3,\\
p_3+p_1p_2-p_1p_2p_3=p'_1p'_2,\\
p_1+p_2+p_3-1=p_1p_2p_3.
\end{cases}
\end{gather*}
\end{Example}

\section{Tetrahedron equation}\label{sec:TE}

The tetrahedron equation firstly was considered by A.~Zamolodchikov~\cite{Zam80} who has constructed its solution in $S$-form. We~consider the following form of the equation
\begin{gather}\label{eq:tetrahedron}
T_{123}T_{145}T_{246}T_{356} = T_{356}T_{246}T_{145}T_{123},
\end{gather}
where $T_{ijk}$ is an operator acting nontrivially in the tensor product of three vector spaces~$V_i$, $V_j$,~$V_k$, indexed by $i$, $j$ and $k$. Tetrahedron equation is the higher order analog of the Yang--Baxter equation. Both equations are examples of $n$-simplex equations~\cite{KST} and play an important role in hypercube combinatorics and higher Bruhat orders. For the complete introduction to the topic see for example~\cite{Serg2}.
In this section we present two proofs of the main theorem of the paper:
\begin{Theorem}\label{maintheorem}
The change of variables~\eqref{Tchange} defines the solution of the tetrahedron equation~\eqref{eq:tetrahedron}.
\end{Theorem}

These two proofs have a lot of common points and ideas, but have the crucial differences in the last stages. It is interesting to compare proofs for the purpose of combining arguments of boundary partition functions and the technique of correlation functions in the Ising--Potts models.

At first in the next subsection we prove that the change of variables~\eqref{Tchange} corresponds to the variables transform in the trigonometric solution of a local Yang--Baxter equation.

\subsection{Local Yang--Baxter equation}

Let us recall that the following change of variables $(t_1,t_2,t_3)\mapsto(t'_1,t'_2,t'_3)$ provides an invariance of the Ising model~\eqref{statsumma} under the star-triangle transformation~\eqref{Tchange}:
\begin{gather}
\begin{cases}
t'_1 =& \sqrt{\dfrac{(t_1+t_2t_3)(t_1t_2t_3+1)}{(t_2+t_1t_3)(t_3+t_1t_2)}},\\[2ex]
t'_2 =& \sqrt{\dfrac{(t_2+t_1t_3)(t_1t_2t_3+1)}{(t_1+t_2t_3)(t_3+t_1t_2)}},\\[2ex]
t'_3 =& \sqrt{\dfrac{(t_3+t_1t_2)(t_1t_2t_3+1)}{(t_1+t_2t_3)(t_2+t_1t_3)}}
\end{cases}\nonumber
\\[2ex]
\qquad \qquad\Updownarrow\nonumber
\\[2ex]
\label{T1change}
\begin{cases}
t'_1t'_2 =& \dfrac{t_1t_2t_3+1}{t_3+t_1t_2},\\[2ex]
t'_2t'_3 =& \dfrac{t_1t_2t_3+1}{t_1+t_2t_3},\\[2ex]
t'_1t'_3 =& \dfrac{t_1t_2t_3+1}{t_2+t_1t_3}.
\end{cases}
\end{gather}
Following~\cite{Kor95}, we~construct orthogonal hyperbolic $3\times 3$ matrices $R_{ij}$ which solve the local Yang--Baxter equation
\begin{gather}
R_{12}(t_3)R_{13}(S(t_2))R_{23}(t_1) = R_{23}(S(t'_1))R_{13}(t'_2)R_{12}(S(t'_3)), \label{LYB}
\end{gather}
where $S(t)$ is the following involution
\begin{gather}\label{Schange}
S(t) = \frac{t-1}{t+1}.
\end{gather}
On the left hand side of~\eqref{LYB} we have
\begin{gather}\label{eq:R1}
R_{12}(t_3) =
\begin{pmatrix}
\mathrm{i}\sinh(\log(t_3))&\cosh(\log(t_3))&0\\
\cosh(\log(t_3))&-\mathrm{i}\sinh(\log(t_3))&0\\
0&0&1
\end{pmatrix}\!,
\\[2ex]
\label{eq:R2}
R_{13}(S(t_2)) =
\begin{pmatrix}
\mathrm{i}\sinh(\log(S(t_2)))&0&\cosh(\log(S(t_2)))\\
0&1&0\\
\cosh(\log(S(t_2)))&0&-\mathrm{i}\sinh(\log(S(t_2)))
\end{pmatrix}\!,
\\[2ex]
\label{eq:R3}
R_{23}(t_1) =
\begin{pmatrix}
1&0&0\\
0&\mathrm{i}\sinh(\log(t_1))&\cosh(\log(t_1))\\
0&\cosh(\log(t_1))&-\mathrm{i}\sinh(\log(t_1))
\end{pmatrix}\!.
\end{gather}

\begin{Theorem}\label{th:LYB}
Matrices~\eqref{eq:R1},~\eqref{eq:R2},~\eqref{eq:R3} together with the rules~\eqref{Tchange},~\eqref{Schange} give a solution of~\eqref{LYB}.
\end{Theorem}

\begin{proof}It can be proved by a straightforward computation. For example let us write down the result of the product on the left hand side
\begin{gather}\label{eq:LYB}
\begin{pmatrix}
\dfrac{t_2\big(t_3^2-1\big)}{t_3(t_2^2-1)}&\dfrac{{\rm i}\big(t_1^2t_2^2t_3^2-t_1^2-t_2^2+t_3^2\big)}{2t_1t_3\big(t_2^2-1\big)} &\dfrac{t_1^2t_2^2t_3^2-t_1^2+t_2^2-t_3^2}{2t_1t_3\big(t_2^2-1\big)}
\\[3ex]
-\dfrac{{\rm i}t_2\big(t_3^2+1\big)}{t_3(t_2^2-1)}&\dfrac{t_1^2t_2^2t_3^2+t_1^2+t_2^2+t_3^2}{2t_1t_3\big(t_2^2-1\big)}
&-\dfrac{{\rm i}\big(t_1^2t_2^2t_3^2+t_1^2-t_2^2-t_3^2\big)}{2t_1t_3\big(t_2^2-1\big)}
\\[3ex]
\dfrac{t_2^2+1}{t_2^2-1}&\dfrac{{\rm i}t_2\big(t_1^2+1\big)}{t_1\big(t_2^2-1\big)}&\dfrac{t_2\big(t_1^2-1\big)}{t_1\big(t_2^2-1\big)}
\end{pmatrix}\!.
\end{gather}
At the first glance the product on the right hand side looks much more cumbersome, but occa\-sionally all terms are simplified and the matrix on the right hand side coincides with~\eqref{eq:LYB}.
\end{proof}

\subsection{Tetrahedron equation, first proof}\label{seq:TE-1}
Let us encode the tetrahedron equation by the Figure~\ref{fig1}.

\begin{figure}[h] \centering
\includegraphics[scale=1]{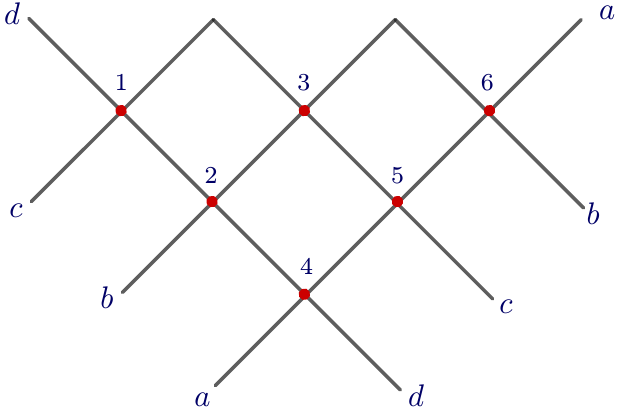}
\caption{Encoding the tetrahedron equation by the standard graph.}\label{fig1}
\end{figure}

The standard graph encodes $R$-matrices in the following way: in each inner vertex numbered by $k$, which is the intersection of strands $i$ and $j$ we put the matrix $R_{ij}(t_k)$ which is the $2\times 2$ matrix in the $4$-dimensional space with basis vectors indexed by $a$, $b$, $c$, $d$. For instance,
\begin{gather*}
R_{ac}(t_5) = \begin{pmatrix}
\mathrm{i}\sinh(\log(t_5))&0&\cosh(\log(t_5))&0\\
0&1&0&0\\
\cosh(\log(t_5))&0&-\mathrm{i}\sinh(\log(t_5))&0\\
0&0&0&1
\end{pmatrix}\!.
\end{gather*}
Let us orient each strand from the left to the right and multiply $R$-matrices in order of the orientation, for instance for the Figure~\ref{fig1} we have the following product of $R$-matrices
\begin{gather}\label{eq:6R}
R_{cd}(t_1)R_{bd}(S(t_2))R_{bc}(t_3)R_{ad}(t_4)R_{ac}(S(t_5))R_{ab}(t_6).
\end{gather}
We note that the orientation defines the product~\eqref{eq:6R} uniquely.

Then let us apply four local Yang--Baxter equations consequently to the inner triangles with vertices numbered $(1,2,3)$, $(1,4,5)$, $(2,4,6)$ and $(3,5,6)$ as on the Figure~\ref{fig2}. As a result, we~have one and the same standard graph as on the Figure~\ref{fig1} rotated by $\pi$.
\begin{figure}[h]\centering
\includegraphics[scale=0.98]{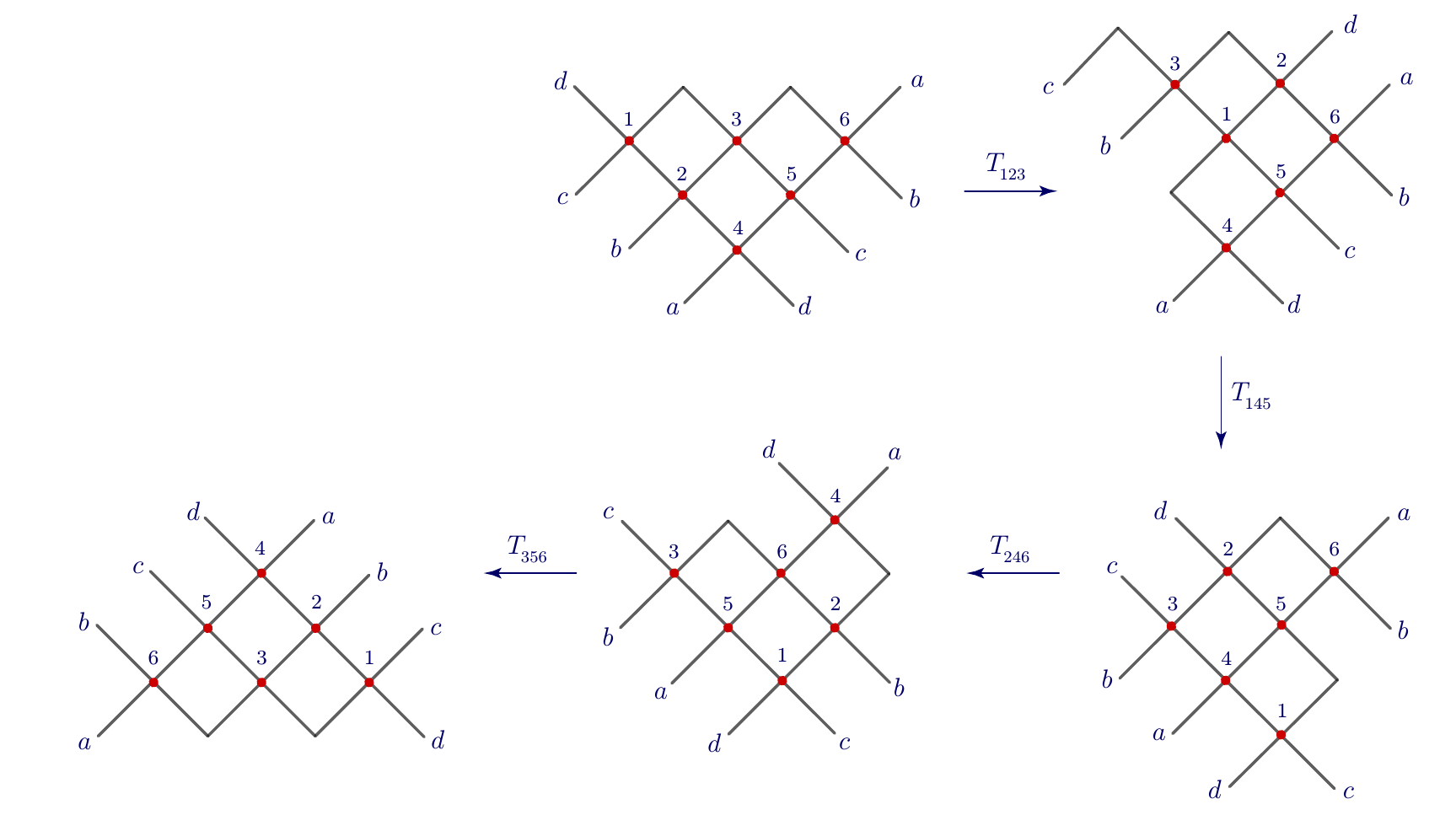}
\caption{Local Yang--Baxter equations applied to the standard graph.}\label{fig2}
\end{figure}

At the same time we could apply local Yang--Baxter equations in the opposite direction: firstly to the triangle $(3,5,6)$, then $(2,4,6)$, $(1,4,5)$ and $(1,2,3)$. Eventually in this case we will have again the same standard graph.

As the reader may have already guessed, every local Yang--Baxter equation applied to the triangle $A$, $B$, $C$ defines the factor $T_{ABC}$ in the tetrahedron equation~\eqref{eq:tetrahedron}. For example we obtain
\begin{gather*}
T_{1,2,3}\colon\ (t_1,S(t_2),t_3,t_4,t_5,t_6)\mapsto(S(t'_1),t'_2,S(t'_3),t_4,t_5,t_6).
\end{gather*}
By Theorem~\ref{th:LYB} the product~\eqref{eq:6R} preserves by each local Yang--Baxter equation encoded on~the Figure~\ref{fig2}. As a result of two sequences of four Local Yang--Baxter equations we obtain an~equality of two products of six $4\times 4$ $R$-matrices
\begin{gather}
R_{cd}(u_1) R_{bd}(u_2) R_{bc}(u_3)R_{ad}(u_4)R_{ac}(u_5)R_{ab}(u_6) \nn
\\ \qquad
{}=R_{cd}(v_1) R_{bd}(v_2) R_{bc}(v_3)R_{ad}(v_4)R_{ac}(v_5)R_{ab}(v_6),
\label{eq:26R}
\end{gather}
where the parameters $u_i$, $v_j$, $i,j=1,\ldots,6$ depend on the initial variables $t_i$, on the mapping~\eqref{Tchange} and on the involution~\eqref{Schange}.

Let us consider this equation element-wise, and note that we could uniquely express para\-me\-ters in the right hand side in terms of the parameters on the left hand side
\begin{gather*}
U_{1,4}=b(t_4),\qquad
U_{2,4}= -a(t_4) b(t_2),\qquad U_{1,3} = a(t_4) b(t_5) ,
\\
U_{1,2} = a(t_4) a(t_5) b(t_6), \quad U_{3,4} = a(t_2) a(t_4) b(t_1),
\quad
U_{2,3} = b(t_2) b(t_4) b(t_5)-a(t_2) a(t_5) b(t_3).
\end{gather*}
Here $U$ is a matrix on the left hand side and $a,b$ are some invertible functions, come from~\eqref{eq:R1}, \eqref{eq:R2}, \eqref{eq:R3}. So we could uniquely determine $t_4$ from the first equation, $t_2$ and $t_5$ from the second and the third, then $t_1$ and $t_6$, and finally $t_3$ from the element $U_{2,3}$.

Let us note that this algebraic proof could be formulated in terms of the paths on the standard graph (Figure~\ref{fig1}) with orientation.
So the equation~\eqref{eq:26R} provides coincidence of the parameters in the vertices given by the two sides of the tetrahedron equation. This finishes the proof.

\subsection{Tetrahedron equation, second proof}\label{sec:tetsecond}
\subsubsection{Involution lemma}
Let us consider the map $F(t_1, t_2, t_3)=(t'_1, t'_2, t'_3)$, where prime variables defined by~\eqref{Tchange}. First of all, we~formulate one technical lemma:

\begin{Lemma} \label{LemmaTech}
	The following identity holds for all $t_1$, $t_2$, $t_3$:
	\begin{gather*}
	S \times S \times S \circ F \circ S \times S \times S= F^{-1},
	\end{gather*}
	where $S(t)=\frac{t-1}{t+1}$.
\end{Lemma}
We present the proof in Appendix~\ref{Ap}.

\subsubsection{Towards the tetrahedron equation}

Let us consider any graph $G$ with a subgraph $\Gamma_1$ which coincides with the leftmost graph on the Figure~\ref{fig5}. We~can transform the graph $G$ to the graph $G'$ with the subgraph $\Gamma_2$ which coincides with the rightmost graph on the Figure~\ref{fig5}. We~could make this mutation by two different chains of star-triangle transformations: $F_{356}^{-1}F_{246}F_{145}^{-1}F_{123}$ and $F_{123}^{-1}F_{145}F_{246}^{-1}F_{356}$. Both are figured out on the Figure~\ref{fig2}. This observation turns us to the following hypothesis
\begin{gather}\label{eq:Zam_Inv}
F_{356}^{-1}F_{246}F_{145}^{-1}F_{123}=F_{123}^{-1}F_{145}F_{246}^{-1}F_{356}.
\end{gather}
\begin{figure}[h]\centering
\includegraphics[scale=1]{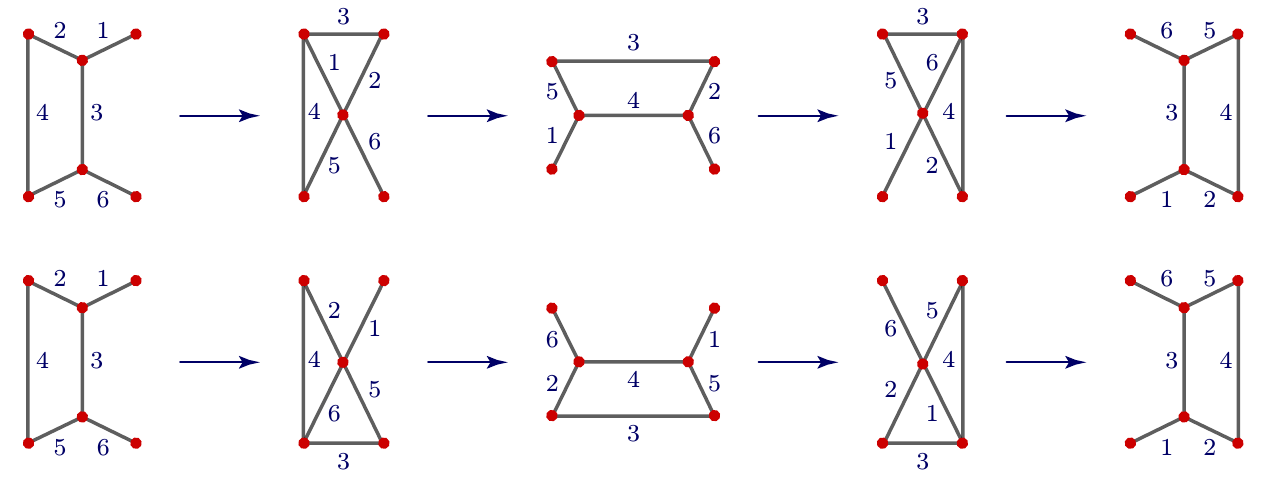}
\caption{The graphical representation of the left and right parts of~\eqref{eq:Phi}.}\label{fig5}
\end{figure}

{\samepage
This equality is equivalent to the Zamolodchikov equation
\begin{gather}
\label{eq:Phi}
\Phi_{356}\Phi_{246}\Phi_{145}\Phi_{123}=\Phi_{123}\Phi_{145}\Phi_{246}\Phi_{356},
\end{gather}
where $\Phi_{ijk}=S_i S_kF_{ijk} S_j$.
Indeed, using Lemma~\ref{LemmaTech} and the simple observation that $S_lF_{ijk}=F_{ijk}S_l$, $l \neq \{i,j,k\}$ we can write down}
\begin{gather*}
F_{356}^{-1}F_{246}F_{145}^{-1}F_{123}=S_3S_5S_6F_{356}S_3S_5S_6F_{246}S_1S_4S_5F_{145}S_1S_4S_5F_{123}
\\ \hphantom{F_{356}^{-1}F_{246}F_{145}^{-1}F_{123}}
{}=S_2S_5(S_3S_6F_{356}S_5)(S_2S_6F_{246}S_4)(S_1S_5F_{145}S_4)(S_1S_3F_{123}S_2)S_2S_5,
\end{gather*}
and
\begin{gather*}
F_{123}^{-1}F_{145}F_{246}^{-1}F_{356}=S_1S_2S_3F_{123}S_1S_2S_3F_{145}S_2S_4S_6F_{246}S_2S_4S_6F_{356}
\\ \hphantom{F_{123}^{-1}F_{145}F_{246}^{-1}F_{356}}
{}=S_2S_5(S_1S_3F_{123}S_2)(S_1S_5F_{145}S_4)(S_2S_6F_{246 }S_4)(S_1S_3F_{356}S_2)S_2S_5.
\end{gather*}
Conjugating both sides of~\eqref{eq:Zam_Inv} by $S_2 S_5$ we obtain the Zamolodchikov tetrahedron equation.

\subsubsection{Solution for the tetrahedron equation}

\begin{Proposition}\label{mainprop}
	The functions
	\begin{itemize}\itemsep=0pt
	\item $\dfrac{\partial \ln(Z(G))}{\partial t_{e}}$,
	where $e$ is any edge belonging to $G - \Omega$, and
	\item $\dfrac{Z_{S(\textbf{A})}(G)}{Z(G)}$, where $Z_{S(\textbf{A})}(G)$ is the boundary partition function and $S$ is any vertex subset of $G - \Omega$ or $G - \Omega'$ $($see the Figure~$\ref{fig startr}),$
	\end{itemize}
are invariant under the star-triangle transformation.
Moreover, these functions do not depend on variables $\beta_{e}$.
\end{Proposition}
\begin{Remark}
The function
\begin{gather*}
\dfrac{Z_{S(\textbf{A})}(G)}{Z(G)}
\end{gather*}
can be interpreted as a probability of the fixed values $\textbf{A}$ of spins in $S$, related to the boundary partition function $Z_{S(\textbf{A})}(G)$.
\end{Remark}

\begin{proof}
The crucial point in the demonstration of the first part of the statement is the fact that the derivative $\frac{\partial \ln(Z(G))}{\partial t_{e_i}}$
does not depend on parameters $\beta$. Indeed, this follows from the explicit form of the partition function
\begin{gather*}
Z(G)=\prod_{e\in E}\beta_e\sum_{\sigma} \prod_{e\in E}(1+(t_e-1)\delta(\sigma_e)).
\end{gather*}
The proof of the second part of the statement is straightforward. It follows from the Definition~\ref{d2} and the condition~\eqref{cond}.
\end{proof}

We will prove the Zamolodchikov equation in its equivalent form
\begin{gather}
\label{eq:Tet}
F_{356}^{-1}F_{246}F_{145}^{-1}F_{123}=F_{123}^{-1}F_{145}F_{246}^{-1}F_{356}.
\end{gather}
Let us notice that due to the local nature of the star-triangle transformation and the convolution property of the boundary partition function we have a choice to take some suitable graph to prove equation~\eqref{eq:Tet}. So let us take the graph $\Gamma_1$ from the Figure~\ref{fig66} with the following choice of boundary set $S_0$:
\begin{gather*}
S_0:=\{v_1, v_2, v_3, v_4\}.
\end{gather*}

\begin{figure}[h!]
 \centering
\includegraphics[scale=1]{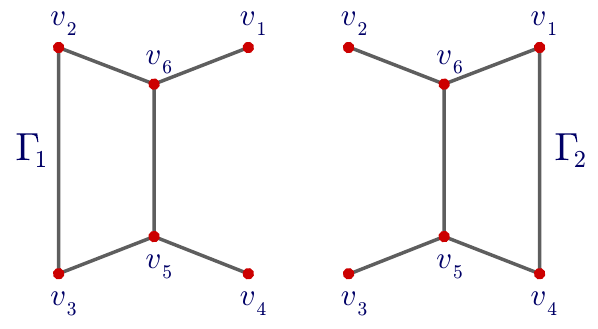}
\caption{The graphs $\Gamma_1$ and $\Gamma_2$.}
\label{fig66}
\end{figure}
We will prove that the values of the second-type invariant functions which are preserved by both sides of the equation~\eqref{eq:Tet} allows us to uniquely reconstruct weights of all edges.
Explain this idea in detail, let us consider the left hand side of the equation~\eqref{eq:Tet} and the map $F_{123}$, then for any $\textbf{A}=\{a_1, \dots,a_4\}$ the following identity holds
\begin{gather*}
\frac{Z_{S_0(\textbf{A})}(\Gamma_1)}{Z(\Gamma_1)}=\frac{Z_{S_1(\textbf{A}_1)}(\Gamma_1)}{Z(\Gamma_1)}+\frac{Z_{S_1(\textbf{A}_2)}(\Gamma_1)}{Z(\Gamma_1)},
\end{gather*}
here $S_1=\{v_1, v_2, v_3, v_4, v_5\}$ (see Figure~\ref{Fig10}), $\textbf{A}_1=\{a_1,\dots,a_4,0\}$, $\textbf{A}_2=\{a_1,\dots,a_4,1\}$.
The Proposition~\ref{mainprop} provides
\begin{gather*}
\frac{Z_{S_1(\textbf{A}_1)}(\Gamma_1)}{Z(\Gamma_1)}=\frac{Z_{S_1(\textbf{A}_1)}(\Gamma')}{Z(\Gamma')},\qquad \frac{Z_{S_1(\textbf{A}_2)}(\Gamma_1)}{Z(\Gamma_1)}=\frac{Z_{S_1(\textbf{A}_2)}(\Gamma')}{Z(\Gamma')},
\end{gather*}
where $\Gamma'$ is obtained from $\Gamma_1$ by the star-triangle transformation (see Figure~\ref{Fig10}). And therefore we deduce that
\begin{gather*}
\frac{Z_{S_0(\textbf{A})}(\Gamma_1)}{Z(\Gamma_1)}=\frac{Z_{S_0(\textbf{A})}(\Gamma')}{Z(\Gamma')}.
\end{gather*}
 Repeating these arguments for the remaining maps $F_{ijk}$ from the left hand side of~\eqref{eq:Tet} we obtain
 \begin{gather}
 \label{system1}
\frac{Z_{S_0(\textbf{A})}(\Gamma_1)}{Z(\Gamma_1)}=\frac{Z_{S_0(\textbf{A})}(\Gamma_2)}{Z(\Gamma_2)}.
 \end{gather}
 In the same fashion, if we consider the right hand side of the equation~\eqref{eq:Tet}, we~similarly obtain that
 \begin{gather*}
\frac{Z_{S_0(\textbf{A})}(\Gamma_1)}{Z(\Gamma_1)}=\frac{Z_{S_0(\textbf{A})}(\Gamma_2)}{Z(\Gamma_2)}.
 \end{gather*}
Hence, in order to prove the equation~\eqref{eq:Tet} it is sufficient to prove that we can reconstruct the parameters $t_i$, $i=1,\ldots,6$ from the values ${Z_{S_0(\textbf{A})}(\Gamma_2)}/{Z(\Gamma_2)}$ for different values of $\textbf{A}$ in a~unique way.
\begin{figure}[h]
\includegraphics[scale=1]{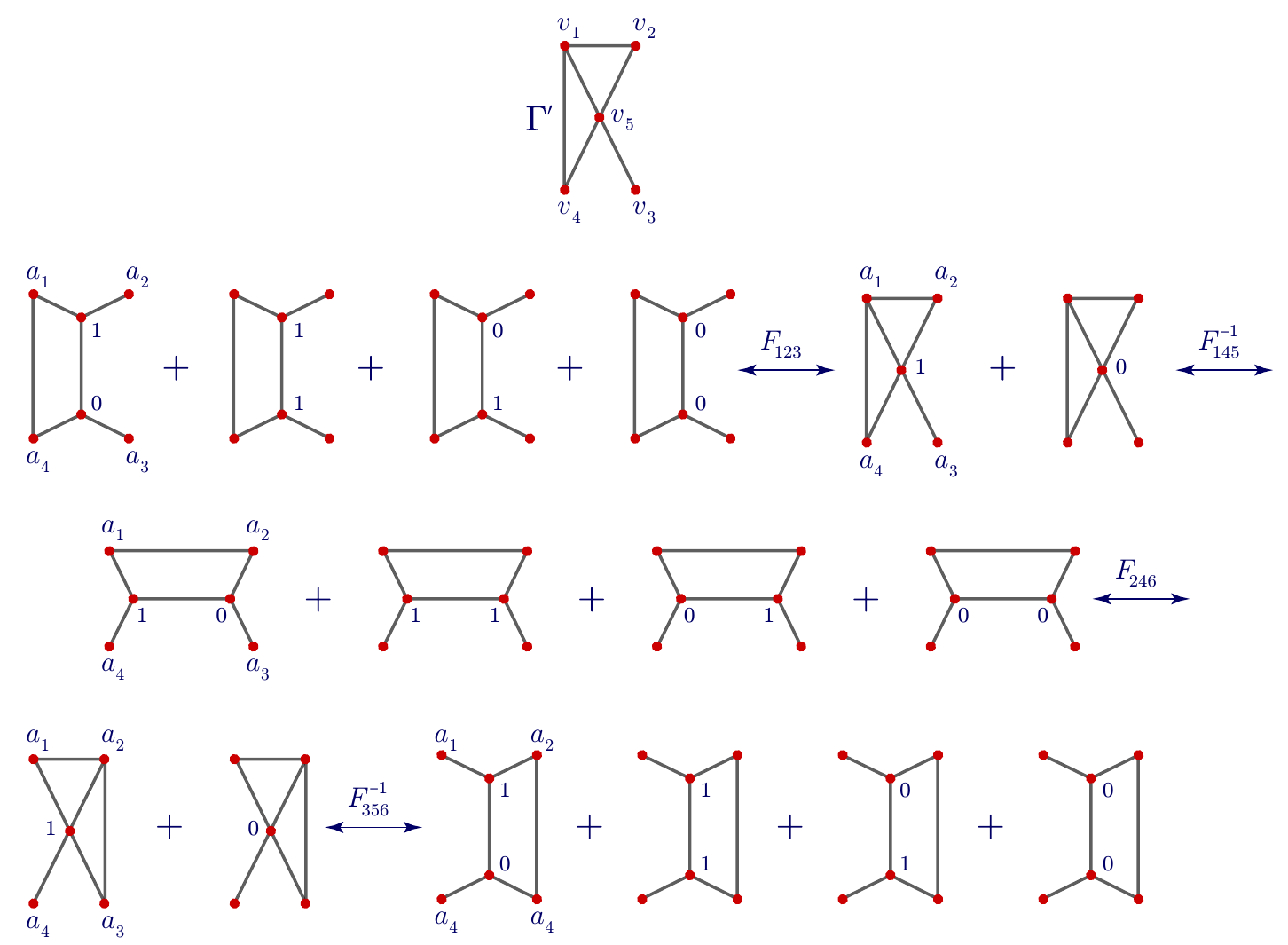}
\caption{The left hand side of~\eqref{eq:Tet}.}
\label{Fig10}
\end{figure}
\begin{figure}[h]
 \centering
\includegraphics[scale=1]{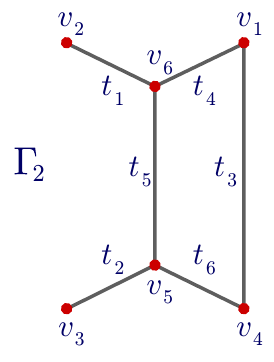}
\caption{The graph $\Gamma_2$.}
\end{figure}

{\sloppy
We understand the identity~\eqref{system1} as a system of $2^4$ linear equations with unknowns $Z_{S_0(\textbf{A})}(\Gamma_2)$ of the following type
\begin{gather*}
Z_{S(\textbf{A})}(\Gamma_2)\colon \quad \frac{Z_{S(\textbf{A})}(\Gamma_2)}{Z(\Gamma_2)}=\alpha(\textbf{A}), \qquad \forall \textbf{A}=(a_1,\ldots,a_4)
\end{gather*}
which is equivalent to
\begin{gather*}
\sum_{\textbf{A}'} Z_{S(\textbf{A}')}(\Gamma_2)= Z_{S(\textbf{A})}(\Gamma_2)/\alpha(\textbf{A}).
\end{gather*}
The rank of the system is equal to 15. Indeed, the rank is $\ge 15$ and we know that there is a~nontrivial solution coming from the boundary partition functions for the graph $\Gamma_2$.}

Hence any solution has the form
 \begin{gather} \label{eq:fin}
 Z_{S_0(\textbf{A})}(\Gamma_2)=C\cdot \alpha_0(a_1,a_2,a_3,a_4),
 \end{gather}
 where $C$ is some constant and $a_i$ are the states. Now we will prove that the parameters $t_1,\ldots,t_6$ are reconstructed uniquely from the equation~\eqref{eq:fin}.

 Let us introduce some auxiliary variables and rewrite the partition function in the following way: we have 16 states of boundary vertices $S_0=\{v_1,v_2,v_3,v_4\}$. Each expression $Z_{S_0(\textbf{A})}(\Gamma_2)$ is a sum of four terms corresponding to the states of internal vertices $v_5$ and $v_6$. We~consider in details the case $S_0=\{0,0,0,0\}$. Let us denote the weights of the states of the square $\{v_1,v_6,v_5,v_4\}$ by $v$, $z$, $y$ and $x$ (Figure~\ref{fig6}). Then, we~obtain the following equations
\begin{alignat}{3}
&v=t_3t_5B,\qquad&& v_1 = t_5t_6B, &\nn
\\
&z=t_3t_4t_5t_6B,\qquad && z_1=t_4t_5B,&\nn
\\
&y=t_6t_3B,\qquad&& y_1=B, &\nn
\\
& x=t_3t_4B,\qquad&& x_1=t_4t_6B,&
\label{eq:t3456}
\end{alignat}
where $B=\beta_3 \beta_4 \beta_5 \beta_6$
and
\begin{gather*}
 v+t_1t_2z+t_2y+t_1x=\frac{C}{B_1}\alpha_0(0, 0, 0, 0), B_1=\beta_1\beta_2.
\end{gather*}
\begin{figure}[h]
 \centering
\includegraphics[scale=1]{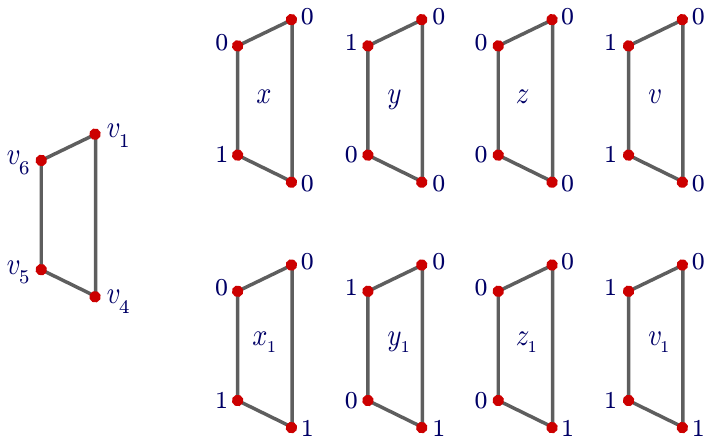}
\caption{The auxiliary variables.}
\label{fig6}
\end{figure}

\noindent
Similarly we obtain seven more equations (we omit eighth equation with $v_1=1$ due to the symmetry of the model with respect to the total involution of spins)
\begin{gather*}
 v_1+t_1t_2z_1+t_2y_1+t_1x_1=\frac{C}{B_1}b_1=\frac{C}{B_1}\alpha'(0, 0, 0, 1), \\
t_1v_1+t_2z_1+y_1t_1t_2+x_1=\frac{C}{B_1}b_2=\frac{C}{B_1}\alpha'(0, 1, 0, 1), \\
t_1t_2v_1+z_1+t_1y_1+t_2x_1=\frac{C}{B_1}b_3=\frac{C}{B_1}\alpha'(0, 1, 1, 1), \\
 t_2v_1+t_1z_1+y_1+t_1t_2x_1=\frac{C}{B_1}b_4=\frac{C}{B_1}\alpha'(0, 0, 1, 1), \\
 v+t_1t_2z+t_2y+t_1x=\frac{C}{B_1}a_1=\frac{C}{B_1}\alpha'(0, 0, 0, 0), \\
t_1v+t_2z+yt_1t_2+x=\frac{C}{B_1}a_2=\frac{C}{B_1}\alpha'(0, 1, 0, 0), \\
t_ 1t_2v+z+t_1y+t_2x=\frac{C}{B_1}a_3=\frac{C}{B_1}\alpha'(0, 1, 1, 0), \\
 t_2v+t_1z+y+t_1t_2x=\frac{C}{B_1}a_4=\frac{C}{B_1}\alpha'(0, 0, 1, 0).
\end{gather*}

 By straightforward calculation we retrieve
\begin{gather}
 y_1=\frac{C}{B_1}\frac{-t_2 b_1 + t_1 t_2 b_2 + b_4 - t_1 b_3}{-t_2^2 + 1 + t_1^2t_2^2 - t_1^2},\qquad
 z_1=\frac{C}{B_1}\frac{t_1 t_2 b_1 - t_2 b_2 + b_3 - t_1 b_4}{-t_2^2 + 1 + t_1^2t_2^2 - t_1^2},\nn
 \\
 x_1=\frac{C}{B_1}\frac{t_2 t_1 b_4 - t_1 b_1 - t_2 b_3 + b_2}{-t_2^2 + 1 + t_1^2t_2^2 - t_1^2},\qquad
 v_1=\frac{C}{B_1}\frac{b_1 + t_1 t_2 b_3 - t_2 b_4 - t_1 b_2}{-t_2^2 + 1 + t_1^2t_2^2 - t_1^2},
\label{eq:xyzv}
\\[2ex]
 y=\frac{C}{B_1}\frac{-t_2 a_1 + t_1 t_2 a_2 + a_4 - t_1 a_3}{-t_2^2 + 1 + t_1^2t_2^2 - t_1^2}, \qquad
 z=\frac{C}{B_1}\frac{t_1 t_2 a_1 - t_2 a_2 + a_3 - t_1 a_4}{-t_2^2 + 1 + t_1^2t_2^2 - t_1^2}, \nn \\
 x=\frac{C}{B_1}\frac{t_2 t_1 a_4 - t_1 a_1 - t_2 a_3 + a_2}{-t_2^2 + 1 + t_1^2t_2^2 - t_1^2}, \qquad
 v=\frac{C}{B_1}\frac{a_1 + t_1 t_2 a_3 - t_2 a_4 - t_1 a_2}{-t_2^2 + 1 + t_1^2t_2^2 - t_1^2}.\nn
\end{gather}
 Using the auxiliary variables it is easy to see that
\begin{gather*}
\frac{z_1}{x_1}=\frac{v}{y}=\frac{-t_2 a_4 + t_1 t_2 a_3 - t_1 a_2 + a_1}{ t_2 t_1 a_2 - t_2 a_1 - t_1 a_3 + a_4}=\frac{t_1 t_2 b_1 - t_2 b_2 + b_3 - t_1 b_4}{t_2 t_1 b_4 - t_1 b_1 - t_2 b_3 + b_2}.
\end{gather*}
In the same fashion we obtain
\begin{gather*}
 \frac{v_1}{x_1}=\frac{v}{x}.
\end{gather*}
Then we can straightforwardly deduce expressions for the variables $t_1$ and $t_2$ (equations~\eqref{t1}, \eqref{t2} in Appendix~\ref{Apb}), then obtain expressions for the auxiliary variables from equations~\eqref{eq:xyzv} and finally obtain variables $t_3$, $t_4$, $t_5$, $t_6$ from the equations~\eqref{eq:t3456}:
\begin{gather*}
 \begin{cases}
 t_3=\sqrt{\dfrac{vyx}{zy^2_1}},\\
 t_5=\dfrac{v}{t_3y_1}, \\
 t_6=\dfrac{y}{t_3y_1},\\
 t_4=\dfrac{x}{t_3y_1}.
 \end{cases}
 \end{gather*}

This completes the proof of the Zamolodchikov equation due to the fact that there is a~uni\-que way to choose positive weights for the edges of the model to provide the expected values of~boundary partitions function for the graph $\Gamma_2$.

\section{Star-triangle transformation, Biggs formula and conclusion}
\label{sec:BST}
The main results of the paper represent the functional relations on the space of multivariate Tutte polynomials. This problem is a step of the program of investigation of the framed graph structures and the related statistical models. We~examined in details the Biggs formula and applied it to the multivariate case. We~also provided a new proof of the theorem of Matiyasevich as a partial case of such formula. The second principal result is the reveal of the tetrahedral symmetry of the multivariate Tutte polynomial at the point $n=2$. Therefore, we~have a~con\-nec\-tion between the multivariate Tutte polynomial, functions on Lustig cluster manifolds~\cite{BFZ} and its electrical analogues~\cite{GT,LP}. We~would like to interpret this property as the critical point of~the model described by the multivariate Tutte polynomial, and the tetrahedral symmetry as a~longstanding analog of the conformal symmetry of the Ising model at the critical point~\cite{Conf}.

Both correspondences are related by the following observation. Let $G$ and $G'$ be two graphs related by the star-triangle transformation. Let us consider the case when the partition function is invariant with respect to this transformation ($n=2$ or $n>2$ and the system~\eqref{eq:q>3},~\eqref{condeq} holds)
\[
Z_n(G')=Z_n(G). 
\]

On the other hand the star-triangle transformation provides a groupoid symmetry on a wide class of objects, in our case on the space of Ising models. The Biggs formula allows us to extend this action to the points of valency $1$ and $2$. And we can obtain the 14-term relation (Theorem~\ref{T:14}) by comparing the right-hand sides of Biggs formulas for $G$ and $G'$.

Let us explain this idea in details, consider two pairs of $n$-Potts models: $M_1(G, i^1_e)$ and $M_1(G', i^1_{e})$, $M_2(G, i^2_{e})$ and $M_2(G', i^2_{e})$ (for simplicity we denote these models $M_1(G)$, $M_2(G)$, $M_1(G')$, $M_2(G')$ correspondingly). After multiplying both parts of the formula~\eqref{eq:aisBiggs} for $M_1(G)$ and $M_2(G)$ by $n^{v(G)}$ (by $n^{v(G')}$ for $M_1(G')$ and $M_2(G')$) we obtain
\begin{gather}
\label{twopart}
Z^1_n(G)=\prod\limits_{e\in G}q_e\sum_{A\subseteq G} \prod\limits_{e\in A} \frac {p_e}{q_e} Z^2_n(A)
=Z^1_n(G')=\prod_{e\in G'}q'_e\sum_{A'\subseteq G'}\prod_{e\in A'} \frac {p'_e}{q'_e} Z^2_n(A'),
\end{gather}
here in both cases we take the sum over the set of all spanning subgraphs.

 Let us rewrite the first part of the formula~\eqref{twopart} by separating two kinds of terms
 \begin{gather} \label{eq:m}
Z^1_n(G)=\prod_{e \in G}q_e\sum\limits_{A_1\subseteq G}\prod_{e \in A_1}\frac{p_e}{q_e} Z^2_n(A_1)+\prod_{e \in G}q_e\sum\limits_{A_2\subseteq G}\prod_{e \in A_2}\frac{p_e}{q_e}Z^2_n(A_2),
 \end{gather}
 where each subgraph $A_1$ contains the full triangle and each $A_2$ contains only a part of the triangle.

 After the star-triangle transformation of the $M_1(G)$ and $M_2(G)$ we obtain the following formula for the models $M_1(G')$ and $M_2(G')$:
 \begin{gather} \label{eq:m'}
Z^1_n(G')=\prod_{e \in G'}q'_e\sum\limits_{A_1'\subseteq G'}\prod_{e \in A'_1}\frac{p'_e}{q'_e} Z^2_n(A_1')+\prod_{e \in G'}q'_e\sum\limits_{A'_2\subseteq G'}\prod_{e \in A'_2}\frac{p'_e}{q'_e}Z^2_n(A_2'),
 \end{gather}
 where each subgraph $A'_1$ contains the full star and each $A'_2$ contains only a part of the star.
 Then, we~compare the terms of these formulas:
 \begin{itemize}\itemsep=0pt
 \item We notice that due to the star-triangle transformation $Z_n^i(G)=Z_n^i(G')$ and $Z_n^i(A_1)=Z_n^i(A_1')$ (here and below $A_1$ is different from $A_1'$ only by the star-triangle transformation).
 \item Also it is easy to see that $\prod_{e \in G}q_e\frac{1}{\prod_{e \in A_1}q_e}=\prod_{e \in G'}q'_{e}\frac{1}{\prod_{e \in A_1'}q'_{e}}$.
 \item If the model $M_1(G)$ is chosen such that $p_1p_2p_3=p_1'p_2'p_3'$, we~conclude that $\prod_{e \in A_1}p_e=\prod_{e \in A'_1}p'_{e}$.
 \end{itemize}
 Now, we~are ready to formulate the following theorem:

\begin{Theorem} \label{T:14}
 Consider two $n$-Potts models $M_2(G)$ and $M_2(G')$, which are different from each other by the star-triangle transformation. Then, the following formula holds
 \begin{gather}
 q_1q_2q_3\bigg(Z^2_n(G_0)+\frac{p_1}{q_1}Z^2_n(G_1)+\frac{p_2}{q_2}Z^2_n(G_2)+\frac{p_3}{q_3}Z^2_n(G_3)
 +\frac{p_1p_2}{q_1q_2}Z^2_n(G_{12})+\frac{p_1p_3}{q_1q_3}Z^2_n(G_{13})\nn
 \\ \hphantom{ q_1q_2q_3\bigg(}
{} + \frac{p_2p_3}{q_2q_3}Z^2_n(G_{23})\bigg)
 =q_1'q_2'q_3'\bigg(Z^2_n(G'_0)+\frac{p'_1}{q'_1}Z^2_n(G'_1)+\frac{p'_2}{q'_2}Z^2_n(G_2)
 +\frac{p'_3}{q'_3}Z^2_n(G'_3)\nn
 \\ \hphantom{ q_1q_2q_3\bigg(}
 +\frac{p'_1p'_2}{q'_1q'_2}Z^2_n(G'_{12}) +\frac{p'_1p'_3}{q'_1q'_3}Z^2_n(G'_{13})
 +\frac{p'_2p'_3}{q'_2q'_3}Z^2_n(G'_{23})\bigg),
 \label{eq:rel}
\end{gather}
where
\begin{gather*}
p_i=\frac{\alpha^1_i-\beta^1_i}{\alpha^2_i-\beta^2_i},\qquad q_i=\frac{\alpha^2_i\beta^1_i-\alpha^1_i\beta^2_i}{\alpha^2_i-\beta^2_i}, \qquad p'_i=\frac{\alpha'^1_i-\beta'^1_i}{\alpha'^2_i-\beta'^2_i}, \qquad q_i'=\frac{\alpha'^2_i\beta'^1_i-\alpha'^1_i\beta'^2_i}{\alpha'^2_i-\beta'^2_i},\nn
\end{gather*}
 variables $\alpha^k_i$, $\beta^k_i$ and $\alpha'^k_i$, $\beta'^k_i $ are related by the star-triangle transformation with the condition $p_1p_2p_3=p_1'p_2'p_3'$ and graphs $G_{i}$ and $G_{ij}$, $i,j=0,1,2,3$ are depicted on the Figure~$\ref{12-term}$.
\end{Theorem}

 \begin{figure}[h]
 \centering
\includegraphics[scale=1]{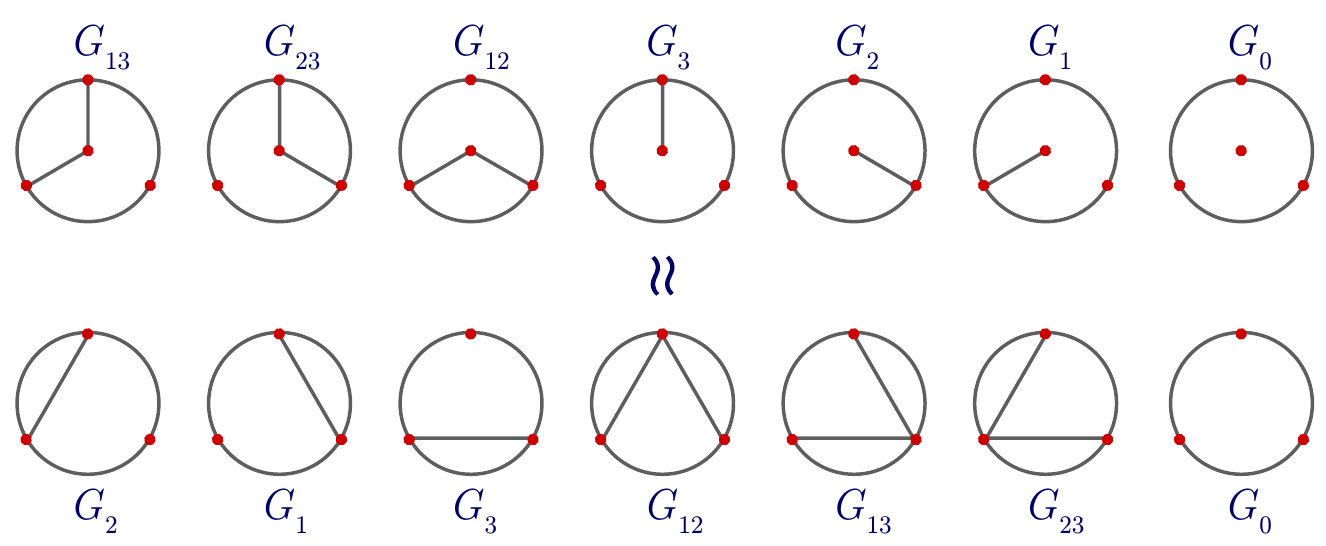}
\caption{The 14-term relation.}
\label{12-term}
\end{figure}
\begin{proof}
We will prove this theorem using induction on $ex(G):=e(G)-3$.

We show that the base of the induction $k=0$ is trivial. Hence, let us consider the $n$-Potts models $M_2(G)$ and $M_2(G')$ and the special models $M_1(G)$ and $M_1(G')$ such that $p_1p_2p_3=p_1'p_2'p_3'$. Then, we~write the formulas~\eqref{eq:m} and~\eqref{eq:m'}, after the comparison for each terms using the reasoning above we immediately obtain the result in the case $ex(G)=0$.

Then, make the step of induction. Again, let us write down the formulas~\eqref{eq:m} and~\eqref{eq:m'}:
\begin{gather*}
Z^1_n(G)=\prod_{e \in G}q_e\sum\limits_{A_1\subseteq G}\prod_{e \in A_1}\frac{p_e}{q_e} Z^2_n(A_1)+\prod_{e \in G}q_e\sum\limits_{A_2\subseteq G}\prod_{e \in A_2}\frac{p_e}{q_e}Z^2_n(A_2)+\frac{\prod_{e \in G}q_e}{q_1q_2q_3}S_1,
 \end{gather*}
 where each $A_1$ contains the full triangle, each $A_2$ contains only a part of the triangle and such that $ex(A_2) \neq ex(G)$, and by $S_1$ we denoted the left hand side of~\eqref{eq:rel},
\begin{gather*}
Z^1_n(G')=\prod_{e \in G'}q_e'\sum\limits_{A_1'\subseteq G'}\prod_{e \in A'_1}\frac{p'_e}{q'_e} Z^2_n(A_1')+\prod_{e \in G'}q_e'\sum\limits_{A'_2\subseteq G'}\prod_{e \in A'_2}\frac{p'_e}{q'_e}Z^2_n(A_2')+\frac{\prod_{e \in G'} q'_e}{q'_1q'_2q'_3}S_2,
 \end{gather*}
 where each $A'_1$ contains the full star, each $A'_2$ contains only a part of the star and such that $ex(A'_2) \neq ex(G')$, and by $S_2$ we denoted the right hand side of~\eqref{eq:rel}.

Then the induction assumption ends the proof.
\end{proof}

We consider these results in the context of numerous generalizations, both for other models of statistical physics, and in a purely mathematical direction. In particular, we~are interested in applying this technique to the Potts model in the presence of an external magnetic field~\cite{EM}, including an inhomogeneous one. In addition, we~are going to develop these methods in a more general algebraic sense, in particular in a non-commutative situation. Partial results of this activity have already been obtained in~\cite{BKT2}.

\appendix
\section{The proof of Lemma~\ref{LemmaTech}} \label{Ap}

\begin{proof}
We start by reformulating this statement in terms of equivalent rational identities. Let us introduce the $x$-variables by the following formula
	\begin{gather}
	F \circ S^3(t_1, t_2, t_3)=(x_1, x_2, x_3),\nn\\
	S^{3}(x_1, x_2, x_3)=(t_1', t_2', t_3').\label{def:X}
	\end{gather}
Here $S^3(t_1,t_2,t_3) = (S\times S\times S)(t_1,t_2,t_3) = (S(t_1),S(t_2),S(t_3))$.
Then the statement of the lemma is equivalent to
	\begin{gather*}
	(t_1', t_2', t_3')= S^{3}(x_1, x_2, x_3)=F^{-1}(t_1, t_2, t_3).
	\end{gather*}
This identity is equivalent to three algebraic relations (we will write down only one of them, because the others differ just by replacing the indices)
\begin{gather}	
t_1t_2=\frac{t_1't_2't_3'+1}{t_3'+t_1't_2'}=
\bigg(\frac{(x_1-1)(x_2-1)(x_3-1)}{(x_1+1)(x_2+1)(x_3+1)}+1\bigg)\bigg/ \bigg(\frac{x_3-1}{x_3+1}+\frac{(x_1-1)(x_2-1)}{(x_1+1)(x_2+1)}\bigg)\nn
\\ \hphantom{t_1t_2}
{}=\frac{x_1+x_2+x_3+x_1x_2x_3}{x_3-x_2-x_1+x_1x_2x_3}
=\frac{(x_1+x_2+x_3+x_1x_2x_3)x_1x_2x_3}{(x_3-x_2-x_1+x_1x_2x_3)x_1x_2x_3}.
\label{eq:TX}
\end{gather}
Now let us introduce some additional variables
\begin{gather*}
t_{12}=x_1x_2,\qquad
t_{23}=x_2x_3,\qquad
t_{13}=x_1x_3,\qquad
a_1=x_1^2,\qquad
a_2=x_2^2,\qquad
a_3=x_3^2.
\end{gather*}
We could rewrite~\eqref{eq:TX} in the following way
\begin{gather*}
t_1t_2=\frac{a_1t_{23}+a_2t_{13}+a_3t_{12}+t_{12} t_{23} t_{13}}
{-a_2t_{13}-a_1t_{23}+a_3t_{12}+t_{12} t_{23} t_{13}}.
\end{gather*}
The equations~\eqref{T1change} and~\eqref{def:X} with the identification $y_i:=S(t_i)$ provide the following system
\begin{gather*}
t_{12}=\frac{y_1y_2y_3+1}{y_3+y_1y_2}=\frac{t_3+t_2+t_1+t_1t_2t_3}{t_3-t_2-t_1+t_1t_2t_3},
\\
t_{13}=\frac{y_1y_2y_3+1}{y_2+y_3y_1}=\frac{t_3+t_2+t_1+t_1t_2t_3}{t_2-t_1-t_3+t_1t_2t_3},
\\
t_{23}=\frac{y_1y_2y_3+1}{y_1+y_3y_2}=\frac{t_3+t_2+t_1+t_1t_2t_3}{t_1-t_2-t_3+t_1t_2t_3},
\\
a_1=\frac{(y_1y_2y_3+1)(y_1+y_2y_3)}{(y_2+y_1y_3)(y_3+y_1y_2)}
=\frac{(t_1-t_2-t_3+t_1t_2t_3)(t_3+t_2+t_1+t_1t_2t_3)}{(t_3-t_2-t_1+t_1t_2t_3) (t_2-t_1-t_3+t_1t_2t_3)}, \\
a_2=\frac{(y_1y_2y_3+1)(y_2+y_1y_3)}{(y_1+y_2y_3)(y_3+y_1y_2)}
=\frac{(t_2-t_1-t_3+t_1t_2t_3)(t_3+t_2+t_1+t_1t_2t_3)}{(t_3-t_2-t_1+t_1t_2t_3) (t_1-t_2-t_3+t_1t_2t_3)}, \\
a_3=\frac{(y_1y_2y_3+1)(y_3+y_2y_1)}{(y_2+y_1y_3)(y_1+y_3y_2)}
=\frac{(t_3-t_2-t_1+t_1t_2t_3)(t_3+t_2+t_1+t_1t_2t_3)}{(t_1-t_2-t_3+t_1t_2t_3) (t_2-t_1-t_3+t_1t_2t_3)}. \end{gather*}
	Using these expressions we can compute
\begin{gather*}
t_{12}a_3+t_{13}a_2+a_1t_{23}+t_{12} t_{23} t_{13}
\\ \qquad
{}=\frac{4(t_3+t_2+t_1+t_1t_2t_3)^2 t_1t_2t_3}{(t_3-t_2-t_1+t_1t_2t_3) (t_2-t_1-t_3+t_1t_2t_3)(t_1-t_2-t_3+t_1t_2t_3)},
\\[2ex]
t_1t_2(-a_2t_{13}-a_1t_{23}+a_3t_{12}+t_{12} t_{23} t_{13})
\\ \qquad
{}=\frac{4(t_3+t_2+t_1+t_1t_2t_3)^2t_1t_2t_3}{(t_3-t_2-t_1+t_1t_2t_3) (t_2-t_1-t_3+t_1t_2t_3)(t_1-t_2-t_3+t_1t_2t_3)}.
\end{gather*}
	In this way we observe that
\begin{gather*}
t_1t_2(-a_2t_{13}-a_1t_{23}+a_3t_{12}+t_{12} t_{23} t_{13})=a_3t_{12}+a_2t_{13}+a_1t_{23}+t_{12} t_{23} t_{13}.
\end{gather*}
This completes the proof.
\end{proof}

\section{For the second proof of the tetrahedron equation} \label{Apb}
In this Appendix we present some technical part of the proof from Section~\ref{sec:tetsecond}. We~present closed formulas for $t_1$ and $t_2$ variables (see the discussion after~\eqref{eq:xyzv})
\begin{gather}
t_1=\big({-}b_3 a_3 + a_2 b_2 + a_1 b_1 - b_4 a_4 +
(b_3^2 a_3^2 - 2 b_3 a_3 a_2 b_2- 2 b_3 a_3 a_1 b_1
 - 2 b_3 a_3 b_4 a_4 \nn
 \\ \hphantom{t_1=(}
 {}+ a_2^2 b_2^2 -
 2 a_2 b_2 a_1 b_1- 2 a_2 b_2 b_4 a_4 + a_1^2 b_1^2 -2 a_1 b_1 b_4 a_4 + b_4^2 a_4^2+
 4 b_4 a_3 a_1 b_2 \nn
 \\ \hphantom{t_1=(}
 {}+ 4 a_2 b_1 b_3 a_4)^{1/2} \big)/(2 (-b_4 a_3 + a_2 b_1)),\label{t1}
\\
 t_2=(a_2b_4-b_2a_4-b_3a_1+a_3b_1+(a_2^2b_4^2-2a_4b_4a_2b_2-2a_2b_4b_3a_1
 -2a_2b_1a_3b_4\nn
 \\ \hphantom{t_2=(}
 {}+b_2^2a_4^2-2b_3a_4a_1b_2-
 -2b_2a_4a_3b_1+b_3^2a_1^2-2a_3b_3a_1b_1+a_3^2b_1^2+4a_3b_4a_1b_2\nn
 \\ \hphantom{t_2=(}
 {}+4a_2b_1b_3a_4)^{1/2})
 /(2(a_3b_4-b_3a_4)).\label{t2}
\end{gather}

\subsection*{Acknowledgements}
We are thankful to V.~Gorbounov for indicating us the strategy of the first proof of the tetrahedron equation in the trigonometric case in Section~\ref{seq:TE-1}. The research was supported by the Russian Science Foundation (project 20-61-46005). The authors thank the anonymous referees for very useful comments which are improved the paper a lot.

\pdfbookmark[1]{References}{ref}
\LastPageEnding

\end{document}